\documentclass[a4paper]{article}
\usepackage[margin=1in]{geometry}
\usepackage{enumerate}
\usepackage{amssymb}
\usepackage{amsthm}
\usepackage{amsmath}
\usepackage{graphicx}
\usepackage{stmaryrd}
\usepackage{url}

\usepackage[ruled,vlined]{algorithm2e}

\usepackage[final,mode=multiuser]{fixme}

\fxuseenvlayout{colorsig}
\fxusetargetlayout{color}
\FXRegisterAuthor{tn}{atn}{TN}
\fxsetface{env}{}

\newcommand{\Occ}{\mathit{Occ}}

\newlength\savedwidth

\newcommand{\LZ}{\mathsf{LZ}}

\newcommand{\SE}[2]{\mathit{SE}_{#1}^{#2}}

\newcommand{\encblock}{\mathit{LC}}
\newcommand{\encpow}[1]{\mathit{RLE}(#1)}
\newcommand{\deltaLR}[1]{\Delta_{#1}}
\newcommand{\assign}{\mathit{Assgn}}

\newcommand{\leave}{\mathit{leave}_{\mathcal{X}}}
\newcommand{\locus}{\mathit{locus}_{\mathcal{X}}}
\newcommand{\pathfunc}{\mathit{path}_{\mathcal{X}}}
\newcommand{\child}{\mathit{child}_{\mathcal{X}}}
\newcommand{\suffixlink}{\mathit{slink}_{\mathcal{X}}}

\newcommand{\cOcc}{\mathit{cOcc}}

\newcommand{\occ}{\mathit{occ}}
\newcommand{\Insert}{\mathit{insert}}
\newcommand{\Delete}{\mathit{delete}}

\newcommand{\dep}{\mathit{dep}}

\newcommand{\code}{\mathit{C}}

\newtheorem{theorem}{Theorem}

\newtheorem{lemma}[theorem]{Lemma}

\newtheorem{observation}[theorem]{Observation}

\newtheorem{example}[theorem]{Example}

\title{A compressed dynamic self-index for highly repetitive text collections}

 \author{
   Takaaki Nishimoto$^1$\quad
   Yoshimasa Takabatake$^2$ and
   Yasuo Tabei$^1$\\
  {$^1$ RIKEN Center for Advanced Intelligence Project, Chuo-ku, Tokyo, Japan}\\
  {\texttt{\{takaaki.nishimoto, yasuo.tabei\}@riken.jp}}\\
  {$^2$ Kyushu Institute of Technology, Fukuoka, Japan}\\
  {\texttt{takabatake@ai.kyutech.ac.jp}}\\
 }

\date{}

\begin{document}
\maketitle

\begin{abstract}
We present a novel compressed dynamic self-index for highly repetitive text collections. Signature encoding, an existing self-index of this type, has a large disadvantage of slow pattern search for short patterns. We obtain faster pattern search by leveraging the idea behind a truncated suffix tree (TST) to develop the first compressed dynamic self-index, called the {\em TST-index}, that supports not only fast pattern search but also dynamic update operations for highly repetitive texts. Experiments with a benchmark dataset show that the pattern search performance of the TST-index is significantly improved.

\end{abstract}

\section{Introduction}~\label{sec:intro}
A {\em highly repetitive text collection} is a set of texts such that any two texts can be converted into each other with a few modifications. Such collections have become common in various fields of research and industry. Examples include genome sequences for the same species, version-controlled documents, and source code repositories. For human genome sequences, it is said that the difference between individual human genomes is around 0.1\%, and there is a huge collection of human genomes, such as the 1000 Genome Project~\cite{genome1000}. As another example, Wikipedia belongs to the category of version-controlled documents. There is clearly a strong, growing need for powerful methods to process huge collections of repetitive texts.

A {\em self-index} is a data representation for a text with support for random access and pattern search operations. Quite a few self-indexes for highly repetitive texts have been proposed thus far. Early methods include the SLP-index~\cite{DBLP:conf/spire/ClaudeN12a} and LZ-end~\cite{DBLP:journals/tcs/KreftN13}. The block tree index (BT-index)~\cite{DBLP:conf/spire/Navarro17} is a recently proposed self-index that reaches compression close to that of Lempel-Ziv. The run-length FM-index~(RLFM-index)~\cite{DBLP:journals/corr/GagieNP17} is a recent breakthrough built on the notion of a run-length-encoded Burrows-Wheeler transform~(BWT). The RLFM-index can be constructed in a space-efficient, online manner while supporting fast operations. Although existing self-indexes for highly repetitive texts are constructible in a compressed space and also support fast operations, there are no proposed methods supporting dynamic updating of self-indexes for editing a set of texts with a highly repetitive structure. Thus, an important open challenge is to develop a self-index for highly repetitive texts that supports not only fast operations for random access and pattern search but also dynamic updating of the index. 

The {\em ESP-index}~\cite{DBLP:conf/wea/TakabatakeTS14} and {\em signature encoding}~\cite{DBLP:conf/stringology/NishimotoIIBT16} are two self-indexes using the notion of a {\em locally consistent parsing~(LCPR)}~\cite{DBLP:journals/algorithmica/MehlhornSU97} for highly repetitive texts. While they have the advantages of being constructible in an online manner and supporting dynamic updates of data structures, they have a large disadvantage of slow pattern search for short patterns. From a given string, these methods build a parse tree that guarantees upper bounds for parsing discrepancies between different appearances of the same substring. For pattern search, the ESP-index performs top-down search of the parse tree to find candidate pattern appearances, and then it checks whether the pattern does occur around each candidate. In contrast, signature encoding uses a 2D range of reporting queries for pattern search. Traversing the search space, especially for short patterns, takes a long time for both methods, which limits their applicability in practice.

\smallskip

In this paper, we present a novel, compressed dynamic index for highly repetitive text collections, called a {\em truncated suffix tree- (TST-) based index (TST-index)}. The TST-index improves on existing self-indexes on an LCPR to support faster pattern search by leveraging the idea behind a {\em $q$-truncated suffix tree~(TST)}~\cite{DBLP:journals/tcs/NaAIP03}. A $q$-TST is built from an input text, and a self-index is created for the transformed text by using the $q$-TST, which greatly improves the search performance. In addition, the TST-index supports dynamic updates, which is a useful property when input texts can be edited, as noted above. Experiments with a benchmark dataset of highly repetitive texts show that pattern search with the TST-index is much faster than that with signature encoding. In addition, the TST-index has pattern search performance competitive with that of the RLFM-index, yet the size of the TST-index can be smaller than that of the RLFM-index.

\section{Preliminaries} \label{sec:preliminary}
Let $\Sigma$ be an ordered alphabet and $\sigma = |\Sigma|$. Let string $T$ be an an element in $\Sigma^*$, with its length denoted by $|T|$. In general, a string $P \in \Sigma^*$ is called a pattern. Let $N=|T|$ and $m=|P|$ throughout this paper. For string $T = xyz$, $x$, $y$, and $z$ are called a prefix, substring, and suffix of $T$, respectively. For two strings $x$ and $y$, let $x \cdot y = xy$. 

The empty string $\varepsilon$ is a string of length $0$. Let $\Sigma^+ = \Sigma^* - \{\varepsilon\}$. For any $1 \leq i \leq |T|$, $T[i]$ denotes the $i$-th character of $T$. For any $1 \leq i \leq j \leq |T|$, $T[i..j]$ denotes the substring of $T$ that begins at position $i$ and ends at position $j$. Let $T[i..] = T[i..|T|]$ and $T[..i] = T[1..i]$ for any $1 \leq i \leq |T|$. For strings $T$ and $K$ and indexes $i$ and $k$, we define insertion and deletion operations as follows: $insert(T,i,K) = T[..i-1]\cdot K \cdot T[i..]$ and $delete(T, i, k)=T[..i-1]\cdot T[i+k-1..]$. For any strings $P$ and $T$, let $\Occ(P,T)$ denote all occurrence positions of $P$ in $T$; that is, $\Occ(P, T) = \{i \mid P = T[i..i+|P|-1], 1 \leq i \leq |T|-|P|+1\}$. Let $occ=|Occ(P,T)|$. Similarly, $\cOcc(c, T) = \Occ(c, T)$ for $c \in \Sigma$.

For a string $P$ and integer $q$, we say that $P$ is a $q$-\emph{gram} if $|P| = q$. We then say that $P$ is a $q$-\emph{short pattern} if $|P| \leq q$, or a $q$-\emph{long pattern} if $|P| > q$. Similarly, we say that a suffix of length at most $q$ of $P$ is a $q$-\emph{short suffix} of $P$. Let $\Sigma_{T}^{q}$ be the set of all substrings of length $q$ and all suffixes of length at most $q$ in $T$, i.e., $\Sigma_{T}^{q} = \{T[i..\min \{i+q-1, |T| \}] \mid i \in 1 \leq i \leq |T| \}$. For example, if $T = {\rm babababbabab}\$$ and $q = 4$, then $\Sigma_{T}^{q} = \{{\rm \$, ab\$, abab, abba, b\$, bab\$, baba, babb, bbab} \}$. For a string $T \in \Sigma^{+}$, we say that $T$ is a \emph{$|\Sigma|$-colored sequence} if $T[i] \neq T[i+1]$ holds for any $1 \leq i < |T|$.

For a string $T$ and an integer $1 < i \leq |T|$, the longest prefix $f$ of $T[i..]$ such that $f$ occurs in $T[..i]$ is called the \emph{longest previous factor without self-reference} at position $i$ of $T$. Furthermore, $f_1, \ldots, f_d$ is called a \emph{factorization} of $T$ if $f_1 \cdots f_d = T$ holds, where factor $f_i$ for each $i\in \{1,2,...,d\}$ is a substring of $T$ and $d$ is the size of the factorization. Then, the run-length encoding $RLE(T)$ is a factorization of string $T$ such that each factor is a maximal run of the same character in $T$. Such a run is denoted as $a^k$ for length $k$ and character $a\in \Sigma$. For example, for $T=aabbbbbabb$, $RLE(T)=a^2b^5a^1b^2$. $RLE(T)$ is a $|RLE(T)|$-colored sequence when we treat each factor as a character.

Next, the \emph{Lempel-Ziv77 (LZ77) factorization without self-reference}~\cite{LZ77} for a string $T$ is a factorization $\LZ(T) = f_1, \ldots, f_z$ satisfying the following conditions: (1) $f_1 \cdots f_z = T$; (2) $f_1 = T[1]$; (3) $f_{i}$ is the longest previous factor without self-reference at position $\ell = |f_{1} \cdots f_{i-1}|+1$ of $T$, if it exists for $1 < i \leq z$; and (4) otherwise, $f_{i} = T[\ell]$. Hence, $z$ is the number of factors in LZ77, i.e., $z = |\LZ(T)|$. For example, if $T = abababcabababcabababcd$, then $\LZ(T) = a,b,ab,ab,c,abababc,abababc,d$.

A self-index is a data structure built on $T$ and supporting the following operations:
\begin{itemize}
\setlength{\itemsep}{-0.2mm} 
\setlength{\parskip}{-0.7mm} 
\item {\bf Count:} Given a pattern $P$, count the number of appearances of $P$ in $T$. 
\item {\bf Locate:} Given a pattern $P$, return all positions at which $P$ appears in $T$. 
\item {\bf Extract:} Given a range $[i..i+ \ell - 1]$, return $T[i..i + \ell - 1 ]$.
\end{itemize}
Such a self-index is considered static. In contrast, a dynamic self-index supports not only the above three operations but also an update operation on the self-index for insertion and deletion of a (sub)string of length $k$ on string $T$ of maximal length $M \geq N$. A compressed dynamic self-index is a compressed representation of a dynamic self-index. We assume that $M = N$ for a static setting and $M \geq N$ for a dynamic setting. 

Our model of computation is a unit-cost word RAM with a machine word size of $W = \Omega(\log_2 M)$ bits. We evaluate the space complexity here in terms of the number of machine words. A bitwise evaluation of space complexity can be obtained with a $\log_2M$ multiplicative factor. 

Finally, let $\log^{(1)} n = \log_{2} n$, $\log^{(i+1)} n = \log \log^{(i)} n$, and $\log^{*} n = \min \{i \mid \log^{i} n \leq 1, i \geq 1 \}$ for a real positive number $n$. 

\section{Literature Review}
Self-indexes for highly repetitive text collections constitute an active research area, and many methods have been proposed thus far. In this section, we review the representative methods summarized in the upper part of Table~\ref{table:resultcomparison1}. See \cite{DBLP:journals/corr/GagieNP17} for a complete survey.

Self-indexes for highly repetitive text collections are broadly classified into three categories. The first category is a grammar-based self-index. Claude and Navarro~\cite{DBLP:conf/spire/ClaudeN12a} presented the {\em SLP-index}, which is built on a {\em straight-line program (SLP)}, a context-free grammar (CFG) in Chomsky normal form for deriving a single text. For the size $n$ of a CFG built from text $T$, the SLP-index takes $O(n)$ space in memory while supporting locate queries in $O(\frac{m^2}{\epsilon} \log (\frac{\log N}{\log n}) + (m + \occ) \log n)$ time, where $\epsilon \in (0,1]$ is a parameter~\cite{DBLP:conf/spire/ClaudeN12a}.

Two other grammar-based self-indexes, the ESP-index~\cite{DBLP:conf/wea/TakabatakeTS14} and signature encoding~\cite{DBLP:conf/stringology/NishimotoIIBT16}, use  the notion of LCPR. Each takes $w = O(z \log N \log^* M)$ space in memory while supporting locate queries in $O(m f_{\mathcal{A}} + \occ \log N + \log w \log m \log^* M (\log N + \log m \log^* M)$ time, where $f_{\mathcal{A}} = f(w, M)$ and $f(a, b) = O(\min \{\frac{\log\log b \log\log a}{\log\log\log b}, \sqrt{\frac{\log a}{\log\log a}} \})$. Signature encoding also supports dynamic updates in $O((k + \log N \\ \log^* M)\log w \log N \log^* M)$ time for highly repetitive texts. Although the ESP-index and signature encoding have an advantage in that they can be built in an online manner, and although signature encoding also supports dynamic updates, these self-indexes have a large disadvantage in that locate queries are slow for short patterns. 

The second category includes self-indexes based on LZ77 factorization. Recently, Bille et al.~\cite{DBLP:conf/cpm/BilleEGV17} presented a self-index with a time-space trade-off on LZ77. Their self-index takes $O(\hat{z} \log (N/\hat{z}))$ space while supporting locate queries in $O(m(1+ \frac{\log^{\epsilon'} \hat{z}}{\log (N/\hat{z})}) + occ(\log \log N + \log^{\epsilon'} \hat{z}))$ time, where $\hat{z}$ is the number of factors in LZ77 with self-reference on $T$ and $\epsilon' \in (0, 1)$ is an arbitrary constant. Navarro~\cite{DBLP:conf/spire/Navarro17} presented a self-index based on a block tree (BT), called the {\em BT-index}. The BT-index uses $O(z\log{(n/z)})$ space and locates a pattern in $O(m^2\log{(N/z)} + m\log^{\epsilon''}{z} +occ(\log\log{N} + \log^{\epsilon''}{z}))$ time for any constant $\epsilon'' > 0$.

The third category includes the RLFM-index built on the notion of a run-length-encoded BWT. The RLFM-index was originally proposed in \cite{Makinen10} but has recently been extended to support locate queries \cite{DBLP:journals/corr/GagieNP17}. It uses $O(r)$ space for the number $r$ of BWT runs and takes $O(m \log \log_{W} (\sigma + N/r) + \occ \log \log_{W} (N/r))$ time for locate queries. Although the RLFM-index supports fast locate queries in a compressed space, it has a serious issue with its inability to dynamically update indexes. 

Despite the importance of compressed dynamic self-indexes for highly repetitive text collections, no previous method has both supported fast queries and dynamic updates of indexes while achieving a high compression ratio for highly repetitive texts. We thus present a compressed dynamic self-index, called the TST-index, that meets both demands and is applicable to highly repetitive text collections in practice. The TST-index has the following theoretical property, which will be proven over the remaining sections. 
\begin{theorem}\label{theo:dynamic_tstesp0}	
	For a string $T$ and an integer $q$, the TST-index takes space $w' = O(z(q^2 + \log N \log^* M))$ while supporting the following four operations: (i) count queries in $O(m (\log\log \sigma)^{2})$ time for a pattern of length $m \leq q$; (ii) locate queries in $O(m (\log \log \sigma)^{2} + occ \log N)$ time for a pattern of length $m \leq q$;(iii) extract queries in $O(\ell + \log N)$ time; and (iv) update operations in
	$O(f_{\mathcal{B}}(k + q + \log N \log^{*} M) + (k+q)q (\log\log
	\sigma)^{2})$ time, where $f_{\mathcal{B}} = f(w', M)$.
\end{theorem}

\begin{table}[t]
	\caption{Summary of self-indexes for highly repetitive text collections. Here, $\occ_{c}
		\geq \occ$ is the number of candidate occurrences of a given pattern as obtained by the
		ESP-index~\cite{DBLP:conf/wea/TakabatakeTS14}. }
	\label{table:resultcomparison1}    
	\tabcolsep = 1mm
	\footnotesize
	\center{	
		\begin{tabular}{r||l|l|l}
			{\bf Method} & {\bf Space} & {\bf Locate Time} & {\bf Update time} \\ \hline
			RLFM-index~\cite{DBLP:journals/corr/GagieNP17} & $O(r)$ & $O(m \log \log_{W} (\sigma + N/r) + \occ \log \log_{W} (N/r))$ & Unsupported \\ \hline
			Bille et al.~\cite{DBLP:conf/cpm/BilleEGV17} & $O(\hat{z} \log (N/\hat{z}))$ & $O(m(1+ \frac{\log^{\epsilon'} \hat{z}}{\log (N/\hat{z})}) + occ(\log \log N + \log^{\epsilon'} \hat{z}))$ & Unsupported \\ \hline	
			BT-index~\cite{DBLP:conf/spire/Navarro17} & $O(z\log{(N/z)})$ & $O(m^2\log{(N/z)} + m\log^{\epsilon''}{z} $ & Unsupported \\
			&  & $+occ(\log\log{N} + \log^{\epsilon''}{z}))$ &   \\ \hline 
			SLP-index~\cite{DBLP:conf/spire/ClaudeN12a} & $O(n)$ & $O(\frac{m^2}{\epsilon} \log (\frac{\log N}{\log n}) + (m + \occ) \log n)$ & Unsupported \\ \hline
			
			Signature & $O(z \log N \log^* M)$ & $ O(m f_{\mathcal{A}} + \mathit{occ} \log{N} + \log{z}$ & $ O((k + \log N \log^* M)\log{z}$ \\ 
			encoding~\cite{DBLP:conf/stringology/NishimotoIIBT16} &  & $\log m \log^* M (\log N + \log m \log^* M)$ &  $\log N \log^* M)$ on average \\ \hline\hline
			TST-index-s & $O(z(q + \log N \log^* N))$ & $ O(m + occ )$~($m \leq q$) & Unsupported \\  
			(this study)                       &                               &   $O( m + occ_{c} \log m \log^{*} N \log N)  (m > q)$    & \\  \hline
			TST-index-d & $O(z(q^2 + \log N \log^* M))$ & $O(m (\log\log \sigma)^{2} + occ \log N)$~($m \leq q$) & $O(f_{\mathcal{B}}(k + q + \log{N}\log^{*} M)$ \\ 
			(this study)                       &                           &                 & $ + (k+q)q (\log\log \sigma)^{2} )$ \\
		\end{tabular}
	}
\end{table}

\section{Fast queries with truncated suffix trees}
In this section we present a novel text transformation, called a {\em $q$-TST transformation}, to improve the search performance of a self-index. We first introduce the TST in the next subsection and then present the $q$-TST transformation in the following subsection. 

\subsection{Tries, compact tries and truncated suffix trees}\label{sec:trie}
A \emph{trie} $\mathcal{X}$ for a set of strings $F$ is a rooted tree whose nodes represent all prefixes of the strings in $F$ (see the left side of Figure~\ref{fig:trie}). Let $U$ be the set of nodes in $\mathcal{X}$, and let $U_{L}$ be the set of leaves in $\mathcal{X}$. We then define the following five operations for $\mathcal{X}$, $u, v
\in U$, $c \in \Sigma$, and $P \in \Sigma^{*}$:
\begin{itemize} \setlength{\itemsep}{-0.0mm} \setlength{\parskip}{-0.0mm}
\item $\pathfunc(u)$: Returns the string $P$ starting at the root and ending at node $u$.
\item $\locus(P)$: Returns the node $u$ such that $\pathfunc(u) = P$.
\item $\leave(P)$: Returns the set of all leaves whose prefixes contain $P$.
\item $\child(u, c)$: Returns the node $v$ such that $\pathfunc(u) \cdot c = \pathfunc(v)$ holds if it exists.
\item $\suffixlink(u) = v$: Returns the node $v$ such that $\pathfunc(v) = \pathfunc(u)[2..]$ holds if it exists.
\end{itemize}

All nodes in $\mathcal{X}$ are categorized into two types. If node $v$ is an internal node and has only one child, then $v$ is called an \emph{implicit node}; otherwise, $v$ is called an \emph{explicit node}. Let $U_{imp}$ and $U_{exp}$ be the respective sets of implicit and explicit nodes in $\mathcal{X}$. 

For $u \in U_{imp}$, the function $\mathit{expl}_{\mathcal{X}}(u)$ returns the lowest explicit node $v \in U_{exp}$ such that $u$ is an ancestor of $v$. The computation times for $\mathit{expl}_{\mathcal{X}}(u)$, $\locus(P)$, $\pathfunc(u) = P$, and $\leave(P)$ are constant, $O(|P|g)$, $O(|P|)$, and $O(|P|g + |\leave(P)|)$, respectively, where $g$ is the computation time for $\child(u, c)$. Here, $g = O(1)$ when we use \emph{perfect hashing}~\cite{DBLP:journals/jacm/FredmanKS84} in $O(|U_{exp}|)$ space. Moreover, $\suffixlink(u)=v$ can be computed in constant time if node $u$ stores a pointer to $v$ in constant space. Hence, the data structure requires $O(|U_{exp}|)$ space. 

A \emph{compact trie} is a space-efficient representation of a trie $\mathcal{X}$ such that all chains of implicit nodes in $\mathcal{X}$ are collapsed into single edges (see the right side of Figure~\ref{fig:trie}). We use two representations for node labels in a compact trie. The first representation uses a string $Y\in \Sigma^*$ such that each edge label in $\mathcal{X}$ is represented as a pair of start and end positions in $Y$, resulting in a space-efficient representation of the edge labels in $\mathcal{X}$. This representation is called a \emph{compact trie with a reference string} and takes $O(|U_{exp}| + |Y|)$ space. The second approach represents each node label explicitly and is called a \emph{compact trie without a reference string}. This representation takes $O(|U|)$ space. The compact trie with a reference string is more space efficient and is used in the static case, while the compact trie without a reference string is used in the dynamic case.

We can insert or delete a string $K$ into or from a compact trie without a reference string in $O(|K|\hat{g})$ time~\cite{DBLP:journals/jacm/Morrison68}, where $\hat{g}$ is the computation time for updating the data structure for $\child(u, c)$ when a child is inserted into or removed from $u$~(assume $g \leq \hat{g}$). Here, $g, \hat{g} = f(u', \sigma)= O((\log \log \sigma)^{2})$ when we use the predecessor/successor approach of Beame and Fich~\cite{DBLP:journals/jcss/BeameF02} in $O(|U_{exp}|)$ space, where $u'$ is the number of children of $u$.
\begin{example}\label{ex:trie}
 The trie on the left side of Figure~\ref{fig:trie} is built on $F = \{{\rm \$, ab\$, abab, abba, b\$, bab\$, baba, babb} , \\{\rm bbab} \}$. In this trie, $U = \{1, \ldots 10, A, \ldots, I\}$, $U_{exp} = \{1, 3, 6, 8, A, \ldots, I \}$, $U_{imp} = {2, 4, 5, 7, 9, 10}$, $\pathfunc(3) = {\rm ab}$, $\locus({\rm baba}) = G$, $\leave(b) = \{E, F, G, H, I \}$, $\child(5, a)
 = D$, $\suffixlink(4) = 7$, and $\mathit{expl}(4) = C$.
\end{example}
\begin{figure}[t]
	\includegraphics[scale=0.4]{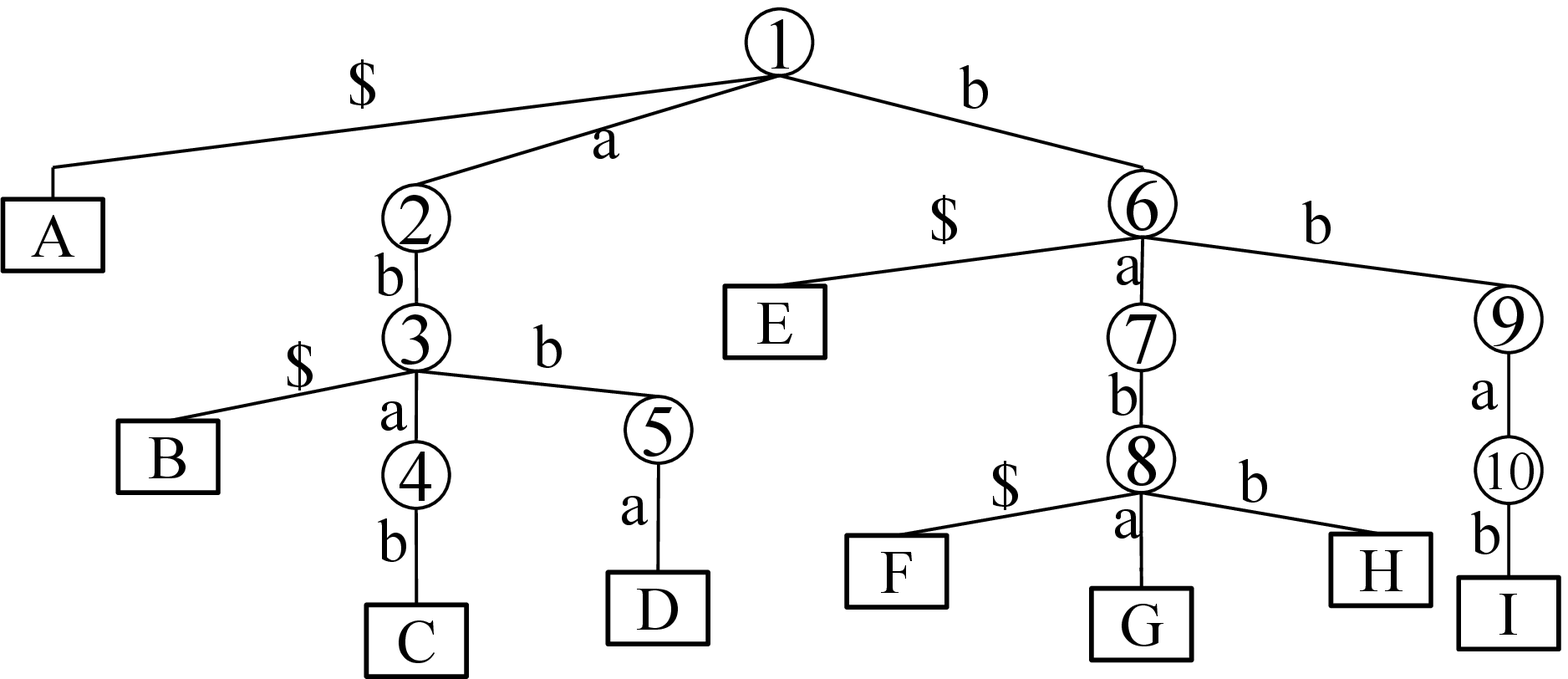}
	\includegraphics[scale=0.4]{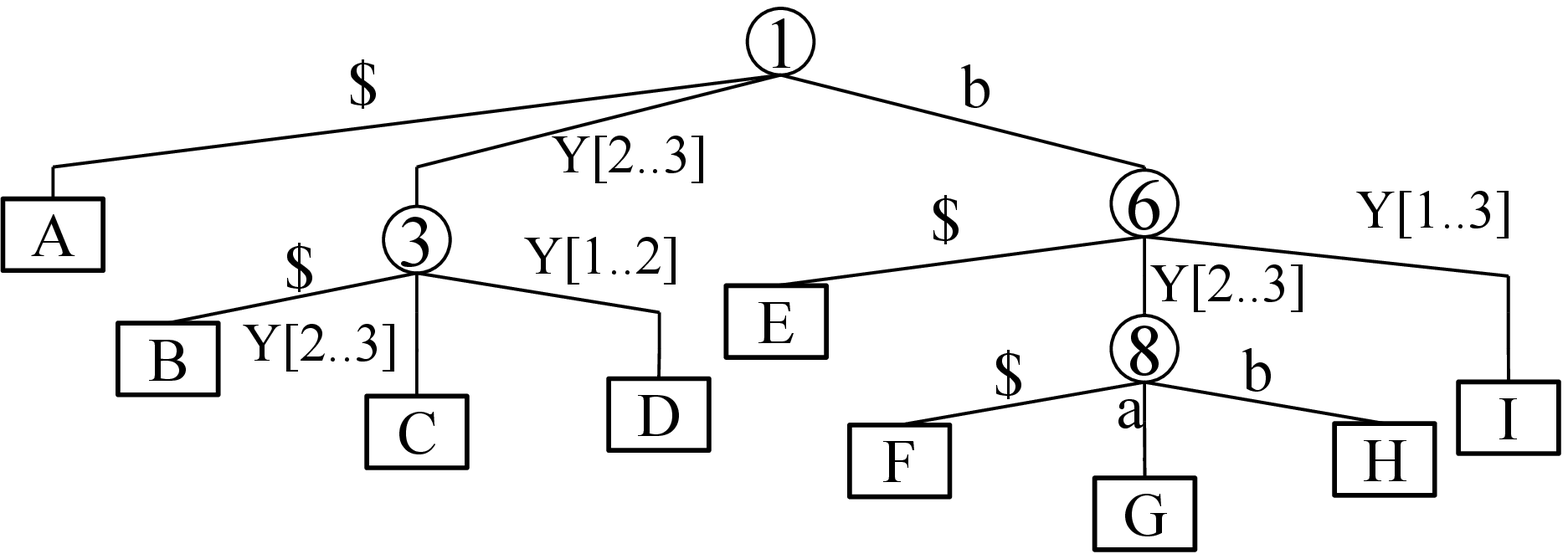}
	\caption{ 
          Trie for $F = \{ {\rm \$,  ab\$, abab, abba, b\$, bab\$, baba, babb, bbab } \}$~(left) and its corresponding compact trie with reference string $Y=bab\$$~(right). 
	}
	\label{fig:trie}
\end{figure}

The $q$-TST~$\mathcal{X}$~\cite{DBLP:journals/tcs/NaAIP03} of a string $T$ is a trie for $\Sigma_{T}^{q}$. We assume that $T$ ends with a special character $\$$ not included in $\Sigma$. Since $\Sigma^q_T=F$ for Example~\ref{ex:trie}, the right side of Figure~\ref{fig:trie} shows a $q$-TST built on $T=babababbabab\$$. Trie $\mathcal{X}$ is a $4$-TST built on $T = {\rm babababbabab}\$$, because $\Sigma_{T}^{q} = F = \{{\rm \$, ab\$, abab, abba, b\$, bab\$, baba, babb, bbab} \}$. 

An important fact is that $\suffixlink(u)$ always exists for every leaf node $u$ in $q$-TST. We thus explicitly store $\suffixlink(u)$ for every leaf $u$. 

Vitale et al. showed that the reference string $Y$ of a $q$-TST can be represented as a string of length $O(|\Sigma_{T}^{q}|)$, and that a $q$-TST for a string $T$ can be constructed in $O(|T| \hat{g})$ time in an online manner while using $O(|\Sigma_{T}^{q}|)$ working space~\cite{VitaleMS15}.

\subsection{$q$-TST transformation}\label{sec:qtrans}
We now present a string transformation using a $q$-TST, which we call a {\em $q$-TST transformation}. A string $T$ is transformed into a new string $T_q$ by a $q$-TST transformation built on $T$. A self-index with improved search performance can then be built on $T_q$.

A $q$-TST transformation using a $q$-TST ${\mathcal X}$ on $T$ is a transformation of $T$ into a new string $T_q$ by replacing every $q$-gram $P$ in $T$ by $\locus(P)$~(i.e., by its node id). Formally, $T_{q} = \code(T, 1) \cdots \code(T, |T|)$, where $\code(T, i) = \locus(T[i..\min \{i+q-1, |T| \}])$. Similarly, a pattern $P$ is transformed into $P_q$ by a $q$-TST transformation using the same $q$-TST ${\mathcal X}$, i.e., $P_q=C(P,1) \cdots C(P,|P|-q+1)$. Given these definitions, the following lemma holds.
\begin{lemma}\label{lem:occ}
	For two strings $T, P$ and $q$-TST of $T$, the following equation holds. 
\begin{eqnarray}\label{Eq:occ}
	\Occ(P, T) &=& 
	\begin{cases}
		\bigcup_{c \in \leave(P)} \cOcc(c, T_{q}) & \mbox{ for } |P| \leq q　　　\\ 
		\Occ(P_{q}, T_{q}) & \mbox{ for } |P| > q　 \\
	\end{cases} 
\end{eqnarray}
For $|P| > q$, $\Occ(P, T) = \phi$ when $P_{q}$ cannot be computed, i.e, when
$P$ contains a new $q$-gram not occurring in $T$.
\end{lemma}
\begin{proof}
Recall that a $q$-gram beginning at position $i$ in $T$ is transformed to the corresponding $q$-TST leaf beginning at position $i$ in $T_{q}$. This means that a substring $P$ of length at most $q$ and beginning at position $i$ in $T$ is transformed to a $q$-TST leaf containing $P$ as a prefix beginning at position $i$ in $T_{q}$. Let $L$ be the set of all leaves containing $P$ as a prefix. Then all occurrence positions of all leaves of $L$ in $T_{q}$ are given by $\Occ(P, T)$. Since $L = \leave(P)$, Equation~\ref{Eq:occ} holds for $q$-short patterns. 

Similarly for a $q$-long pattern $P$, a substring $P$ beginning at position $i$ in $T$ is transformed to $P_{q}$ beginning at position $i$ in $T_{q}$. $\Occ(P, T) = \phi$ clearly holds when $P_{q}$ cannot be computed. Therefore, Equation~\ref{Eq:occ} also holds for $q$-long patterns. 
\end{proof}

\begin{example}\label{ex:tstindex}
	Let $T = {\rm babababbabab}\$$ and $q=4$, and let the trie in Figure~\ref{fig:trie} be a $q$-TST of $T$. Then $T_{q} = {\rm GCGCHDIGCFBEA}$. Let $P = {\rm ab}$ and $P' = {\rm babab}$. Then $\Occ(P, T) = \cOcc(B, T_{q}) \cup \cOcc(C, T_{q}) \cup \cOcc(D, T_{q})$ holds, because $\leave(P) = \{B, C, D \}$, $\cOcc(B, T_{q}) = \{11 \}$, $\cOcc(C, T_{q}) = \{2, 4, 9 \}$, and $\cOcc(D, T_{q}) = \{6 \}$. Similarly, $\Occ(P', T) = \Occ(P'_{q}, T_{q}) = \Occ(GC, T_{q}) = \{1, 3, 8\}$, because $|P'| > q$.
\end{example}
We can compute $\leave(P)$ in $O(|P|g + |\leave(P)|)$ time. $P_{q}$ also can be computed in $O(|P|g)$ time by using $\suffixlink$ and $\child$. This is because for $u = \code(P, i+1)$, $v = \code(P, i)$, and $q \leq i < m$, $u = \child(\suffixlink(v), P[i+q])$ holds if $u$ and $v$ exist. 


Lemma~\ref{lem:occ} tells us that we can improve search queries on a general self-index for $q$-short patterns on a $q$-TST in $O(|P|g + c_1 \times |\leave(P)|)$ time if the computation time for pattern search on the general self-index is greater than $O(|P|g + c_1 \times |\leave(P)|)$, where $c_{1}$ is the computation time for $\cOcc$ with the general self-index. If the length of pattern $P$ is at most $q$, i.e., $|P| \leq q$, then we perform search queries on the $q$-TST. Otherwise, we perform search queries of $P_{q}$ on a self-index for $T_{q}$. 

We obtain the following theorem by using Lemma~\ref{lem:occ}.
\begin{theorem} \label{theo:TST_index}
Let $\mathcal{I}(T)$ be an index supporting $\cOcc(P, T)$ in $O(c_{1})$ time and $\Occ(P, T)$ in $O(c_{2})$ time. Then there exists an index of $O(|\Sigma_{T}^{q}| + |\mathcal{I}(T_{q})| )$ space that supports $\Occ(P, T)$ in $O(m + c_{1} \times |\leave(P)|)$ time for a $q$-short pattern $P$, and that supports $\Occ(P, T)$ in $O(m + c_{2})$ time for a $q$-long pattern $P$, where $|\mathcal{I}(T)|$ is the size of $\mathcal{I}(T)$.
\end{theorem}
We apply Theorem~\ref{theo:TST_index} to signature encoding and present a novel self-index named TST-index in the next section.

\section{TST-index}\label{sec:se}
We obtain the TST-index by combining Theorem~\ref{theo:TST_index} with signature encoding. First, we introduce LCPR and signature encoding, and then we develop the TST-index.

\subsection{Locally consistent parsing~(LCPR)}
LCPR~\cite{DBLP:journals/algorithmica/MehlhornSU97} is a factorization of a string
$T$ by using a bit string $\tau(T)$ computed for $T$. Let $p_i$ be the position of the
$i$-th $1$ in $\tau(T)$. Then $\tau(T)$ and $p_i$ satisfy the conditions given in
Lemma~\ref{lem:lc}.
\begin{lemma}[\cite{DBLP:journals/algorithmica/MehlhornSU97}]\label{lem:lc}
  For a $c$-colored sequence $T$ of length at least $2$, there exists a function $\tau(T)$ that returns a bit sequence of length $|T|$ satisfying the following properties: (1) $2 \leq p_{i+1} - p_{i} \leq 4$ and $p_{1}=1$ hold for $1 \leq i \leq d$, where $d = |\cOcc(1, \tau(T))|$ and $p_{d+1} = |T| + 1$. (2) If $T[i-\Delta_{L}..i+\Delta_{R}] = T'[j-\Delta_{L}..j+\Delta_{R}]$ holds for integers $i$ and $j$, then $\tau(T)[i] = \tau(T')[j]$ holds, where $\Delta_{L} = \log^* c + 6, \Delta_{R} = 4$, and $T$ and $T'$ are $c$-colored sequences.
\end{lemma}
An LCPR for a $c$-colored string $T$ using $p_i$ for $1 \leq i \leq d+1$ is defined as
$\encblock_{c}(T) = T[p_{1}..p_{2}-1], \ldots, T[p_{d}..p_{d+1}-1]$.

\begin{example} \label{ex:Encblock}
  Let $\deltaLR{L} = 2, \deltaLR{R} = 1$, and $T = {\rm abcabcdabcab}$, and assume that $\tau(T) = 100100010010$. Then $\encblock_{|\Sigma|}(T) = \rm{abc, abcd, abc, ab}$, $p_{1} = 1$, $p_{2} = 4$, $p_{3} = 8$, $p_{4} = 11$, $p_{5} = 13$, and $2 \leq p_{i+1}-p_{i} \leq 4$ holds for any $1 \leq i \leq 4$. Since $T[3- \deltaLR{L}..3+\deltaLR{R}] = T[10-\deltaLR{L}..10+\deltaLR{R}] = {\rm abca}$ holds, $\tau(T)[3] = \tau(T)[10] = 0$. Similarly, since $T[4- \deltaLR{L}..4+\deltaLR{R}] = T[11-\deltaLR{L}..11+\deltaLR{R}] = {\rm bcab}$ holds, $\tau(T)[3] = \tau(T)[10] = 1$.
\end{example}

\subsection{Signature encoding}\label{subsec:se}
A signature encoding~\cite{DBLP:journals/algorithmica/MehlhornSU97} of a string $T$ is a context-free grammar $\mathcal{G} = (\Sigma, \mathcal{V}, \mathcal{D}, S)$ generating the single text $T$ and built using $\mathit{RLE}(T)$ and $\mathit{LC}_{4M}$. Here, $\mathcal{V} = \{e_1, \ldots, e_{w} \}$ is a set of positive integers called \emph{variables}. $\mathcal{D} = \{e_i \rightarrow f_i\}_{i = 1}^{w}$ is a set of \emph{deterministic production rules} (a.k.a. \emph{assignments}), with each $f_i$ being either of a sequence of variables in $V$ or a single character in $\Sigma$. $S$ is the start symbol deriving string $T$.

A signature encoding corresponds to a balanced derivation tree of height $O(\log{N})$ built on the given string $T$, where each internal node (respectively, leaf node) has a variable in $V$ (respectively, a character in $\Sigma$), as illustrated in Figure~\ref{fig:se}. Since this derivation tree is balanced, the node labels at each level can be considered as a sequence from the leftmost node to the rightmost node at a level. We define $\assign^{+}(f_{1}, f_{2}, \ldots , f_{d})$ as a function returning a variable sequence $e_{1}, e_{2}, \ldots, e_{d}$, where $f_i$ is a single character in $\Sigma$ or a sequence of variables in $V$, and $e_i$ is a single character in $\Sigma$. In addition, $(e_{j} \rightarrow f_{j}) \in \mathcal{D}$ holds for $1 \leq j \leq d$. Then, $\SE{t}{T}$ is a sequence of node labels at the $t$-th level of the derivation tree built
from string $T$ and defined using $\assign^{+}(f_{1}, f_{2}, \ldots , f_{d})$ as follows.

\begin{eqnarray*}
  \SE{t}{T} &=&
  \begin{cases}
    \assign^{+}(T[1], \ldots, T[|T|]) & \mbox{ for } t = 0, \\
    \assign^{+}(\encblock_{4M}(\SE{t-1}{T}))& \mbox{ for }  t = 2, 4, \ldots, \\
    \assign^{+}(\encpow{\SE{t-1}{T}}) & \mbox{ for } t = 1, 3, \ldots. \\
    \end{cases} \\
\end{eqnarray*}
Here, $S= \SE{h}{T}[1]$ holds for the minimum positive integer $h$ satisfying $|\SE{h}{T}| = 1$. Hence, $h+1$ is the height of the derivation tree of $S$. Let $w = |\mathcal{V}|$ be the size of $\mathcal{G}$, whose bound is given by the following lemma.
\begin{lemma}[\cite{18045}]\label{lem:upperbound_signature}
	The size $w$ of $G$ is bounded by $O(\min(z \log N \log^* M, N))$.
\end{lemma}

$\mathcal{G}$ satisfies the following properties by the definition. (1) $h = O(\log N)$ holds because $|\SE{t}{T}| \leq \frac{1}{2}|\SE{t-1}{T}|$ for a positive even integer $t$. (2) Every variable is at most $4M$ because $\SE{t}{T}$ must be a $4M$- colored sequence to apply an LCPR to $\SE{t}{T}$. (3) $\mathcal{G}$ can be stored in $O(w)$ space, because every variable $e$ in $\mathcal{D}$ is in one of three cases: (i) $e \rightarrow c \in \Sigma$; (ii) $e \rightarrow \hat{e}^k$, where $\hat{e} \in \mathcal{V}$ and $1 \leq k \leq N$; or (iii) $e$ derives a variable sequence of length $2 \leq d \leq 4$.

\begin{figure}[t]
  \centerline{
    \includegraphics[scale=0.7]{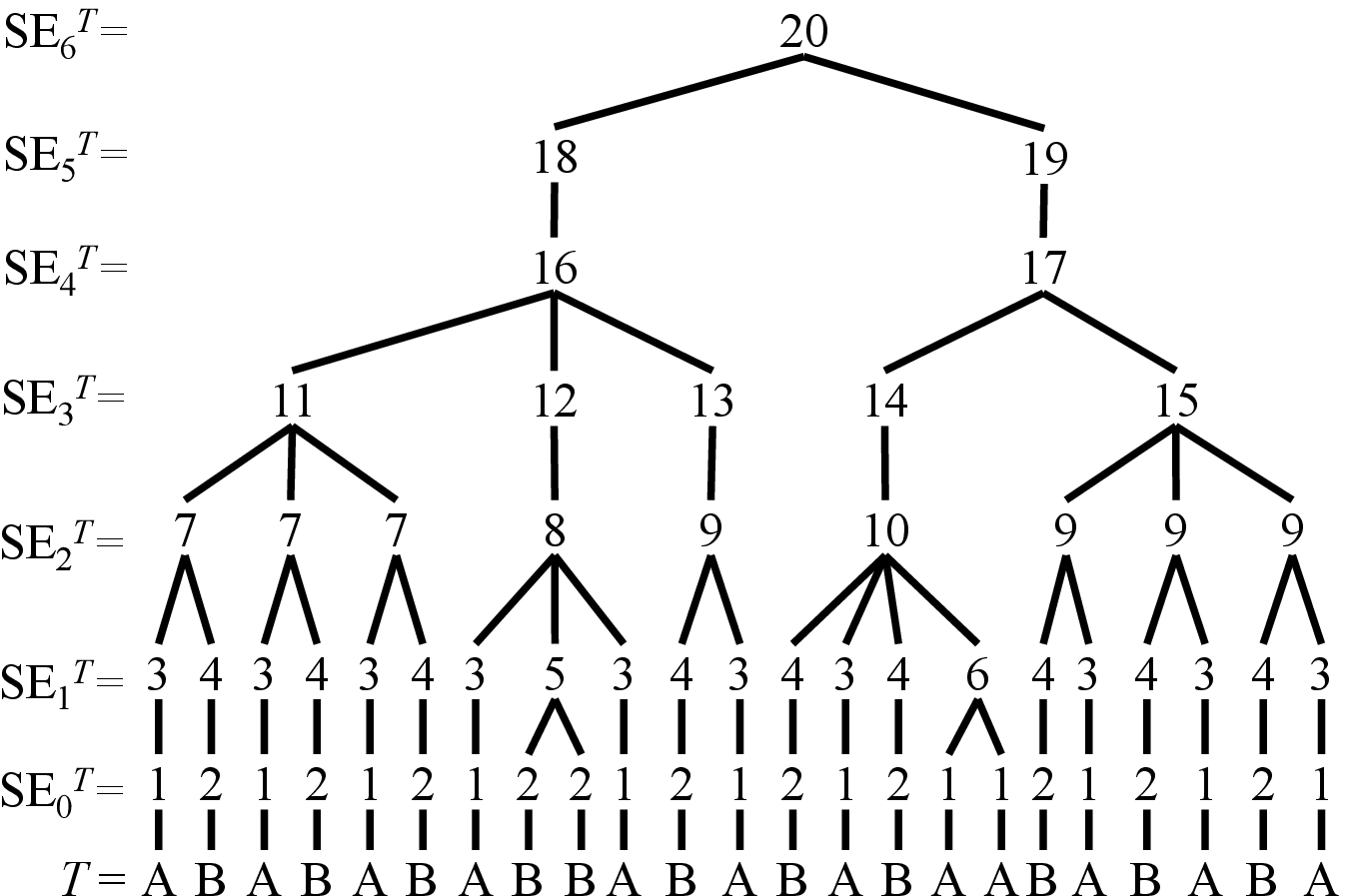}
    }
  \caption{ The derivation tree of $\mathcal{G}$ of Example~\ref{ex:se}.  }
  \label{fig:se}
\end{figure}
\begin{example}\label{ex:se}
  Let $\mathcal{G} = (\Sigma, \mathcal{V}, \mathcal{D}, S)$ be a context-free
  grammar, where $\Sigma = \{A, B\}$, $\mathcal{V} = \{1, \ldots , 20 \}$,
  $\mathcal{D} = \{1 \rightarrow A, 2 \rightarrow B, 3 \rightarrow 1^1, 4 \rightarrow
  	2^1, 5 \rightarrow 2^2, 6 \rightarrow 1^2, 7 \rightarrow (3,4), 8 \rightarrow (3,5, 3), 9
  	\rightarrow (4,3), 10 \rightarrow (4,3,4,6), 11 \rightarrow 7^3, 12 \rightarrow 8^1, 13
  	\rightarrow 9^1, 14 \rightarrow 10^1, 15 \rightarrow 9^3, 16 \rightarrow (11,12,13), 17
  	\rightarrow (14,15), 18 \rightarrow 16^1, 19 \rightarrow 17^1, 20 \rightarrow (18,19)
  	\}$, and $S = 20$. Assume that $\encblock_{4N}(\SE{1}{T}) =
  (3,4)^3,(3,5,3),(4,3),(4,3,4,6), \\ (4,3)^3$, $\encblock_{4N}(\SE{3}{T}) = (11,12,13),
  (14,15)$, and $\encblock_{4N}(\SE{5}{T}) = (18,19)$ hold. Then $\mathcal{G}$ is
  the signature encoding of $T = {\rm ABABABABBAABABABAABABABA}$.
\end{example}

\subsection{Static TST-indexes}
We can extend signature encoding as a self-index for pattern search. A key idea for pattern search is that there exists a common variable sequence of length $O(\log |P| \log^* M)$ representing every substring $P$ in the parse tree of $\mathcal{G}$. This sequence is called the \emph{core} of $P$ and is computable from $P$, as described elsewhere \cite{DBLP:conf/wea/TakabatakeTS14, DBLP:conf/stringology/NishimotoIIBT16} in detail.

Since we can perform top-down searches of the parse tree for the core of $P$ instead of
$P$ itself, we obtain the following lemma.~\footnote{Takabatake et al.~\cite{DBLP:conf/wea/TakabatakeTS14} construct the grammar representing an input text by using a technique called
	\emph{alphabet reduction}, which is essentially equal to LCPR. Although the definition
	of their grammar is slightly different from ours, we can use their search algorithm with
	our signature encoding.} 
\begin{lemma}[\cite{DBLP:conf/wea/TakabatakeTS14}]\label{lem:iesp}
  For a signature encoding, there exists an index of $O(w)$ space supporting count and locate queries in $O(m + occ_{c} \log m \log^{*} M \log N)$ time, $\cOcc(c, T)$ in $O(\occ)$ time, and extract queries in $O(\ell + \log N)$ time, where $\occ_{c} \geq
  \occ$ is the number of candidate occurrences of $P$ by this index.
\end{lemma}
\begin{proof}
	In the original paper of Takabatake et al.~\cite{DBLP:conf/wea/TakabatakeTS14}, their index performs count and locate queries in $O(\log \log w(m\log^{*} M + occ_{c} \log m \log^{*} M \log N))$ time and $\cOcc(c, T)$ in $O(occ \log N)$ time. We can remove the $\log \log w$ term, however, by using perfect hashing~\cite{DBLP:journals/jacm/FredmanKS84}. Since the $m \log^{*} M$ term is the computation time for $\tau(P)$, it can also be reduced to $m$ by using the data structure proposed by Alstrup et al.~\cite{LongAlstrup}. Both of these data structures use $O(w)$ space.
	
	Next, we show that we can compute $\cOcc(c, T)$ in $O(\occ)$ time by using $O(w)$ space. The directed acyclic graph (DAG) of $\mathcal{D}$ represents the derivation tree $\mathcal{T}$ of $\mathcal{G}$, and each node in the DAG represents a distinct variable in $\mathcal{V}$. For an edge $p$ in the DAG, let the \emph{relevant offset} between the parent $u$ and child $v$ be the length of the string generated by $e_{1} \cdots e_{i-1}$, where $u$ represents $e \rightarrow e_{1} \cdots e_{k} \in \mathcal{D}$ and $p$ represents the $i$-th edge in $u$. Then the occurrence position of a node in $\mathcal{T}$ is the sum of relevant offsets in the path starting at the root and ending at the node. Hence, we can compute $\cOcc(c, T)$ in $O(occ \log N)$ time by storing all relevant offsets, because the height of $\mathcal{T}$ is $O(\log N)$. Furthermore, the offsets can be stored in $O(w)$ space, because we can represent the relevant offsets between $e$ and $\hat{e}$ for $e \rightarrow \hat{e}^{k} \in \mathcal{D}$ by using $O(1)$ space.
	
	Next, for a node $v$ representing $e \in \mathcal{V}$, we store the lowest ancestor $u$ of $v$ such that the in-degree of $u$ is at least 2 or $u$ is the root in the DAG, i.e, the node representing $S$. We also store the sum of relevant offsets on the path starting at $u$ and ending at $v$. By summing up these values, we can reduce $O(occ \log N)$ to $O(occ)$ time, and the data structures use $O(w)$ space.
\end{proof}
We call the signature encoding built on $T_{q}$ a $q$-\emph{signature encoding}.
Our static self-index consists of the $q$-TST $\mathcal{X}$ of $T$ and the self-index of Lemma~\ref{lem:iesp} for $\mathcal{G}$ representing $T_{q}$. In turn, we construct $T_{q}$ from $\mathcal{X}$, $\mathcal{G}$ from $T_{q}$, and the self- index of Lemma~\ref{lem:iesp} from $\mathcal{G}$. The search algorithm is based on Theorem~\ref{theo:TST_index}, meaning that, for locate queries, we compute
$\locus(P)$ by using $\mathcal{X}$ and output all occurrences of $\locus(P)$ on
$T_{q}$ by using locate queries on $T_{q}$ if $|P| \leq q$. Otherwise, we compute
$P_{q}$ by using $\mathcal{X}$ and output all occurrences of $P_{q}$ on $T_{q}$ by
using locate queries on $T_{q}$. Similarly, we can perform count queries. Hence, we
obtain the following performance for our self-index from
Theorem~\ref{theo:TST_index} and Lemma~\ref{lem:iesp}.
\begin{theorem}\label{theo:tstesp}
For a string $T$ and an integer $q$, there exists an index using $O(|\Sigma_{T}^{q}| + w')$ space and supporting (i) locate queries in $O(m + \occ)$ time for $q$-short patterns, (ii) locate and count queries in $O( m + occ_{c} \log m \log^{*} M \log N)$ time for $q$-long patterns, and (iii) extract queries in $O(\ell + \log N)$ time, where $w'$ is the size of the $q$-signature encoding of $T$.
\end{theorem}

We give the following upper bound for $w^\prime$.
\begin{theorem}\label{thorem:tight_upper_bound}
  For a $q$-signature encoding of $T$, $w' = O(z(q + \log N \log^* M))$ holds.
\end{theorem}
\begin{proof}
  See Section~\ref{sec:space}.
\end{proof}
We also give an upper bound $|\Sigma_{T}^{q}| = O(zq)$~\cite{DBLP:journals/corr/TanimuraNBIT17}. Thus, the size of our index is $O(z(q + \log N \log^* M))$. The results of an empirical evaluation of the TST-index are given in Section~\ref{sec:exp}.

\subsection{Dynamic TST-indexes}
In this section, we consider maintaining our self-index for $q$-short patterns with a dynamic text $T$. Recall that our index consists of a $q$-TST $\mathcal{X}$ of $T$ and an index of $T_{q}$. In the dynamic setting, the index is still based on Theorem~\ref{theo:TST_index}, but we use the following data structure for $T_{q}$ instead of that given by Lemma~\ref{lem:iesp}.
\begin{lemma}[Dynamic signature encoding~\cite{DBLP:conf/mfcs/NishimotoIIBT16}]\label{lem:dse}
  With the addition of data structures taking $O(w)$ space, $\mathcal{G}$ can support update operations in $O(f_{\mathcal{A}}(k + \log N \log^{*} M))$ time. The modified data structure also supports computing $\cOcc(c, T)$ in $O(\occ \log N)$ time for $c \in \Sigma$ and extract queries in $O(\ell + \log N)$ time.
\end{lemma}
Hence, our dynamic index supports locate queries in $O(mg + occ \log N)$ time for $q$-short patterns. For count queries, we append $|\Occ(\pathfunc(v), T)|$ to $v \in U_{exp}$ as extra information. We also append $\suffixlink(v)$ to $v \in U_{L}$.

The remaining problem is how to update these data structures when $T$ is changed. The main problem is how $T_{q}$ changes when $T$ is changed. If we can understand that process, then we can update the data structures through each update operation.

Let $\code(T, i, j) = \code(T, i) \cdots \code(T, j)$ for a string $T$ and two integers $i$ and $j$, i.e., $\code(T, i, j)=T_{q}[i..j]$. For strings $x, y, z \in \Sigma^{*}$, the following equations hold, where $s = x[|x|-q+1]z[1..q]$ and $s' = x[|x|-q+1]\cdot y \cdot z[1..q]$:

\begin{eqnarray}\label{Eq:concat}
  \code(xz, 1, |xz|) &=& \code(x, 1, |x| - q + 1) \code(s, 1, |s| - q + 1) \code(z, 1, |z|) \nonumber　 \\
  \code(xyz, 1, |xyz|) &=&  \code(x, 1,  |x| - q + 1) \code(s', 1, |s'| - q + 1) \code(z, 1, |z|) 　
\end{eqnarray}

We can explain the changes in $T_{q}$ due to update operations by using Equation~\ref{Eq:concat}. First, $T_{q} = \code(xz, 1, |xz|)$ changes to $\code(xyz, 1, |xyz|)$ when we update $T$ by $\Insert(T, i, K)$, where $x = T[1..i-1]$, $y = K$, and $z = T[i..|T|]$. Similarly, $T_{q} = \code(xyz, 1, |xyz|)$ changes to $\code(xz, 1, |xz|)$ when we update $T$ by $\Delete(T, i, k)$, where $x = T[1..i-1]$, $y = T[i..i+k-1]$, and $z = T[i+k..|T|]$. A substring of length at most $|y| + 2q$ only changes in the $q$-TST leaf sequence of $xz$ when we insert $y$ between $x$ and $z$. This means that we can update $T_{q}$ with two insertions and deletions when $T$ is updated by an insertion or deletion.

Next, we consider the update algorithm for our self-index. For $\Insert(T, i, K)$, we update $\mathcal{X}$ and $\mathcal{G}$ by the following four steps:
\begin{description}
  \item[(i)] Insert all $q$-grams in $s'$ into $\mathcal{X}$.
  \item[(ii)] Compute $\code(s', 1, |s'| - q + 1)$ by using $\mathcal{X}$.
  \item[(iii)] Insert $\code(s', 1, |s'| - q + 1)$ into $T_{q}$, and remove $\code(s, 1, |s| - q + 1)$ from $T_{q}$.
  \item[(iv)] Remove old~(unused) $q$-grams from $\mathcal{X}$.
\end{description}
Similarly, we update the $q$-TST and $\mathcal{G}$ for $\Delete(T, i, k)$. Step (iii) can run in $O(f_{\mathcal{B}}(k + q + \log N \log^{*} M))$ time by Lemma~\ref{lem:dse}. Note that we can detect old $q$-grams by checking $V$, because $\mathcal{G}$ removes old variables during updating.

Next, we consider how to update the $q$-TST $\mathcal{X}$ of $T$. If $\mathcal{X}$ does not have extra information, then it can be updated in $O((k+q)q\hat{g})$ time. Hence, we need to show that $\mathcal{X}$ can maintain extra information without increasing the update time. When a new $q$-gram $P$ is created in $T$, we insert $P$ into $\mathcal{X}$ and increment the value representing $|\Occ(\pathfunc(v), T)|$ for each explicit node $v$ on the path from the root to $u = \locus(P)$. This runs in $O(|P| \hat{g})$ time. If $u$ is a node created by this insertion, then we need to compute $\suffixlink(u)$. We can also do this in $O(|P| \hat{g})$ time by computing $\locus(P[2..])$. Note that $\locus(P[2..])$ always exists after computing Step 1 in the update algorithm. Similarly, we can update the extra information in the same time when a $q$-gram is removed from $T$. Hence Steps (i), (ii), and (iv) can run in $O((k+q)q\hat{g}) $ time.

The only remaining point of discussion is maintenance of leaf IDs in the $q$-TST. When a new leaf $v$ is created in $\mathcal{X}$, we need to assign an unused integer to $v$. In the dynamic setting, we use the address representing leaf $v$ as this integer, so every character in $T_{q}$ uses $W$ bits.

Since $|U| = O(q |\Sigma_{T}^{q}| ) = O(zq^2)$, we obtain Theorem~\ref{theo:dynamic_tstesp0} from Lemma~\ref{lem:dse} and Theorem~\ref{thorem:tight_upper_bound}. Note that our data structure can support extract queries in $O(\ell + \log N)$ time by using Lemma~\ref{lem:dse}.

\section{Tight upper bound for signature encodings of $T_{q}$ }\label{sec:space}
In this section, we obtain the proof of Theorem~\ref{thorem:tight_upper_bound} by using the proof of Lemma~\ref{lem:upperbound_signature}. Note that we can simply show that $w' = O(zq \log N \log^* M)$ by using the following lemma and Lemma~\ref{lem:upperbound_signature}, but this order is larger than that in Theorem~\ref{thorem:tight_upper_bound}.

\begin{lemma}
	$z' = O(zq)$ holds for a string $T$ and integer $q$, where $z' = \LZ(T_{q})$.
\end{lemma}
\begin{proof}
	Let $LPF(i)$ and $LPF_{q}(i)$ be the longest previous factors without self-reference at position $i$ in $T$ and $T_{q}$, respectively. Then $|LPF_{q}(i)| \geq \max \{|LPF(i)| - (q-1), 0 \}$ holds for $1 \leq i \leq N$ by a $q$-TST transformation. This means that every LZ77 factor of $\LZ77(T)$ is divided into at most $q$ LZ77 factors on $T_{q}$. Hence, $z' = O(zq)$.	
	\end{proof}
\subsection{The proof of Lemma~\ref{lem:upperbound_signature}}
We begin by explaining our approach for the proof of Lemma~\ref{lem:upperbound_signature}. Since $|V|$ is equal to $X$, the number of distinct variables in the derivation tree of $\mathcal{G}$, we try to bound $X$ through the following idea: (1) Since each variable is determined by a local context (substring) of $T$ obtained by signature encoding, the variables for common local contexts are the same. (2) Let $\LZ(T) = f_{1}, \ldots f_{z}$. Since each factor is a longest previous factor in $T$, a variable whose local context is contained in an LZ77 factor also occurs at another position. Hence, the upper bound of $X$ is $Y$, the number of variables whose local context is on two or more LZ77 factors of $T$. (3) Finally, since there exist $O(\log N \log^* M)$ such variables in a factor, we can obtain Lemma~\ref{lem:upperbound_signature}.

We can now work through our approach more formally. Every variable in $\mathcal{G}$ is determined by a strictly local context of $T$. To be precise, every variable for $t=0$ is determined only by the represented character. For a positive even integer $t$, every variable is determined by the $\SE{t}{T}$ used by the bit sequence representing the form of the variable. Recall Example~\ref{ex:Encblock}. The factor representing $abcd$ does not change as long as $g(T)[4..8]$ does not change. We can compute $g(T)[4..8]$ from $T[4-\Delta_{L}..8+\Delta_{R}]$, so the factor depends locally on $T[4-\Delta_{L}..8+\Delta_{R}]$. For a positive odd integer $t$, every variable is determined by a run representing $e$ and the left and right characters at either end of the run. Recall the example of $\encpow{T}$. The factor $b^5$ changes into $b^6$ if the left character is $b$, but the factor always remains $b^5$ as long as the characters at both ends do not change. Hence, the factor depends locally on $T[2..8]$. In short, a variable on $\SE{t}{T}$ depends on an interval on $\SE{t-1}{T}$. By recursively applying this rule, a variable ultimately depends on an interval on $T$. For a variable on $\SE{t}{T}[i]$, let $\dep_{t}^{T}(i)$ be the corresponding interval on $T$ on which it depends. Formally, let $\SE{t}{T}[i]$ represent $\SE{t- 1}{T}[L..R]$ when $t$ is a positive even integer, and $\SE{t-1}{T}[L'..R']$ when $t$ is a positive odd integer. Then, we define $\dep_{t}^{T}(i)$ as follows.
\begin{eqnarray*}
	\dep_{t}^{T}(i) &=&
	\begin{cases}
		\{ i \} & \mbox{ for } t = 0, \\
		\bigcup_{x = L - \Delta_{L}}^{R+1+\Delta_{R}} \dep_{t-1}^{T}(x) & \mbox{ for }  t = 2, 4, \ldots, \\
		\bigcup_{x = L' - 1}^{R' + 1} \dep_{t-1}^{T}(x) & \mbox{ for } t = 1, 3, \ldots, \\
	\end{cases} \\
\end{eqnarray*}
\begin{example}
	Recall Example~\ref{ex:se}. Let $\Delta_{L}$ = 2 and $\Delta_{R} = 1$. Then $\dep_{0}^{T}[7] = \{7 \}$, $\dep_{1}^{T}[8] = \{7, \ldots, 10 \}$, and $\dep_{2}^{T}[5] = \dep_{1}^{T}[10-\Delta_{L}] \cup \cdots \cup \dep_{1}^{T}[12+\Delta_{R}] = \{7, \ldots, 15 \}$.
\end{example}
The local context of $\SE{t}{T}[i]$ is a substring on $\dep_{t}^{T}(i)$ in $T$. Note
that each $\dep_{t}^{T}(i)$ represents an interval on $[1, N]$. The following
observation holds with respect to signature encodings and $\dep$.
\begin{observation}\label{ob:depend}
	For two integers $i < j$ and an integer $t$, let $\dep_{t}^{T}(i) = [\ell, r]$ and
	$\dep_{t}^{T}(j) =[\ell', r']$. Then $\ell < \ell'$ and $r < r'$ hold.
\end{observation}
\begin{observation}\label{ob:depend2}
	Let $[\ell, r]$ and $[\ell', r' ]$ be two intervals on $[1, N]$ such that $T[\ell..r] =
	T[\ell'..r']$ holds. If there exists two integers $i$ and $t$ such that $\dep_{t}^{T}(i) =
	[\ell, r]$, then there exists an integer $j$ such that $\dep_{t}^{T}(j) = [\ell', r']$.
\end{observation}
Hence, variables having a common local context are the same.

Next, let $K(x, t) = \{i \mid 1 \leq i \leq |\SE{t}{T}|, x \in \dep_{t}^{T}(i) \}$, constituting the set of variables on $\SE{t}{T}$ that depend on position $x$ in $T$. We obtain the following inequality from Sentence (2) and the
local context property:

 \begin{eqnarray}\label{ksize}
 w = |V| \leq \sum_{i = 1}^{z} \sum_{t = 0}^{h} |K(x_{i}, t)|
 \end{eqnarray}
 holds, where $x_{i} = |f_{1}\cdots f_{i}|$. We can then bound $|K(x, t)|$ by using the
 following lemma.
\begin{lemma}[\cite{LongAlstrup}]\label{lem:depend3}
	(1) For a positive even integer $t$, if $|K(x, t-1)| = d$ holds, then $|K(x, t)| \leq (d +
	\Delta_{L} + \Delta_{R}) / 2$ also holds. (2) For a positive odd integer $t$, if $|K(x, t-
	1)| = d$ holds, then $|K(x, t)| \leq d$ also holds.
\end{lemma}
\begin{proof}
	By Observation~\ref{ob:depend}, we can represent $K(x, t-1)$ by an interval $E = [ p, p+d-1 ]$. Let $X$ be the set of positions $\SE{t}{T}$ whose variable depends on a position in $E$ on $\SE{t}{T}$. Then $X = K(x, t)$ holds by $\mathit{dep}$. Let $\SE{t}{T}[i]$ represent $\SE{t-1}{T}[L_{i}..R_{i}]$. Then $|X| \leq (d + \Delta_{L} + \Delta_{R}) / 2$ holds, because $\SE{t}{T}[i]$ depends on $\SE{t-1}{T}[L_{i}- \Delta_{L}..R_{i}+\Delta_{R}+1]$ and every factor $\encblock_{4M}(\SE{t- 1}{T})$ has a length of at least $2$. (2) Similarly, for a positive odd integer $t$, if $|K(x, t-1)| = d$ holds, then $|K(x, t)| \leq d$ also holds. 
\end{proof}
Hence, $|K(x, t)| = O(\log^* M)$ holds for any integers $x$ and $t$ by Lemma~\ref{lem:depend3} and $|K(x, 1)| \leq 2$. Therefore, Lemma~\ref{lem:upperbound_signature} holds by $h = O(\log N)$, Lemma~\ref{lem:depend3}, and Inequality~\ref{ksize} , where $h + 1$ is the height of the derivation tree of $\mathcal{G}$.

\subsection{The proof of Theorem~\ref{thorem:tight_upper_bound}}
We can now easily prove Theorem~\ref{thorem:tight_upper_bound} by using the same approach for the proof of Lemma~\ref{lem:upperbound_signature}. The point is that we can regard a character at $i$ in $T_{q}$ as depending on $T[i..i+q-1]$. The means that we can regard a variable on $\SE{t}{T_{q}}$ as depending on a local context in $T$, not $T_{q}$. Formally, we modify $\dep_{t}^{T_{q}}(i)$ for a $q$-signature encoding as follows.

\begin{eqnarray*}
	\dep_{t}^{T_{q}}(i) &=&
	\begin{cases}
		[ i, i+q-1 ] & \mbox{ for } t = 0, \\
		\bigcup_{x = L}^{R} \dep_{t-1}^{T_{q}}(x) & \mbox{ for }  t = 2, 4, \ldots, \\
		\bigcup_{x = L' - 1}^{R' + 1} \dep_{t-1}^{T_{q}}(x) & \mbox{ for } t = 1, 3, \ldots, \\
	\end{cases} \\
\end{eqnarray*}
Similarly, $K(x, t)$ is redefined by the modified version of $\dep_{t}^{T_{q}}(i)$. Note that Inequality~\ref{ksize} and Lemma~\ref{lem:depend3} still hold. Hence, $\sum_{t = 0}^{h'} |K(i, t)| = O(q + h' \log^{*} M)$ immediately holds by $|K(x, 0)| \leq q$, where $h' + 1$ is the height of the derivation tree of the signature encoding representing $T_{q}$. Therefore, Theorem~\ref{thorem:tight_upper_bound} holds by $h' = O(\log N)$, where $z$ is the number of LZ77 factors of $T$, not $T_{q}$.

\section{Experiments}\label{sec:exp}
In this section, we demonstrate the effectiveness of the TST-index in a static setting with a benchmark dataset of highly repetitive texts.

For a benchmark dataset we used nine highly repetitive texts consisting of the files DNA, english.200MB, einstein.en.txt, einstein.de.txt, Escherichia\_Coli, cere, influenza, para, and world\_leaders from the Pizza \& Chili corpus (\url{http://pizzachili.dcc.uchile.cl}). We sampled $1000$ substrings of each length $m=\{4, 8, 16, \dots, 2048\}$ from each benchmark text and used those substrings as queries. We used the memory consumption and search time for count and locate queries as evaluation measures. We performed all the experiments on one core of a quad-core Intel(R) Xeon(R) E5-2680 v2 (2.80 GHz) CPU with 256 GB of memory. 

We compared our TST-index with the ESP-index~\cite{DBLP:conf/wea/TakabatakeTS14} and the RLFM-index~\cite{DBLP:journals/corr/GagieNP17}. The ESP-index provides a baseline for evaluating the effectiveness of the TST-index, while the RLFM-index is a state-of-the-art self-index for highly repetitive text collections. We used the C++ language to implement the TST-index~\footnote{\url{https://github.com/TNishimoto/TSTESP}} as a combination of a $q$-TST and the ESP-index, the self-index on an LCPR given in Lemma~\ref{lem:iesp}. We varied the parameter $q$ by testing $q$-gram lengths from \{4,8,16,32\}. We used existing implementations of the ESP-index (\url{https://github.com/tkbtkysms/esp-index-I}) and the RLFM-index (\url{https://github.com/nicolaprezza/r-index}).

\subsection{Results}

\begin{table}[t]
	\caption{Index size in megabytes for each text. The TST-index is denoted as $q$-TST
		for each value of parameter $q$ in $\{4, 8, 16, 32\}$.}
	\label{table:exp1}    
	\tabcolsep = 1mm
	\center{
		\begin{tabular}{r|r|r|r|r|r|r|r|r|r}
			 & DNA & english & einstein & einstein & Escherichia & cere & influenza & para & world \\
			 &         & 200MB & en.txt & de.txt & Coli &  &  & & leaders \\ \hline 
			 Text & 385 & 200 & 445 & 88 & 107 & 439 & 147 & 409 & 44 \\ \hline \hline
			 ESP & 438 & 248 & 2 & 1 & 42 & 47 & 23 & 60 & 5 \\ 
			 RLFM & 2,429 & 760 & 3 & 1 & 146 & 124 & 31 & 164 & 6 \\ \hline \hline
			 $4$-TST & 501 & 388 & 8 & 3 & 48 & 54 & 24 & 71 & 13 \\ 
			 $8$-TST & 826 & 1,987 & 27 & 10 & 87 & 94 & 37 & 124 & 47 \\ 
			 $16$-TST & 23,927 & 10,615 & 48 & 17 & 1,794 & 1,417 & 277 & 1,960 & 88 \\ 
			 $32$-TST & 32,287 & 13,199 & 69 & 24 & 2,292 & 1,876 & 762 & 2,835 & 155 \\ 
		\end{tabular}

			}
		\end{table}
\begin{figure}[h]
	\centerline{
		\includegraphics[scale=0.4]{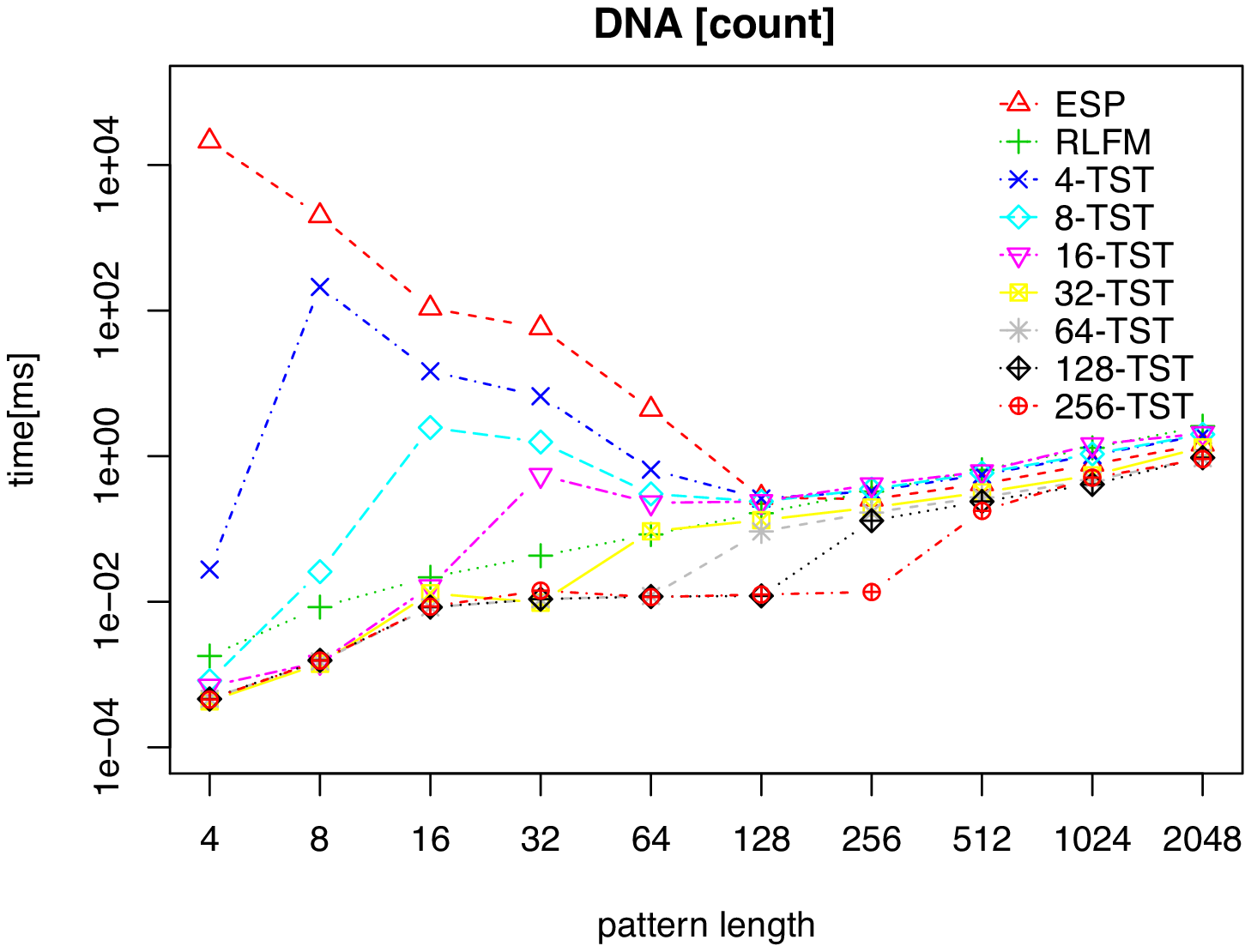}
		\includegraphics[scale=0.4]{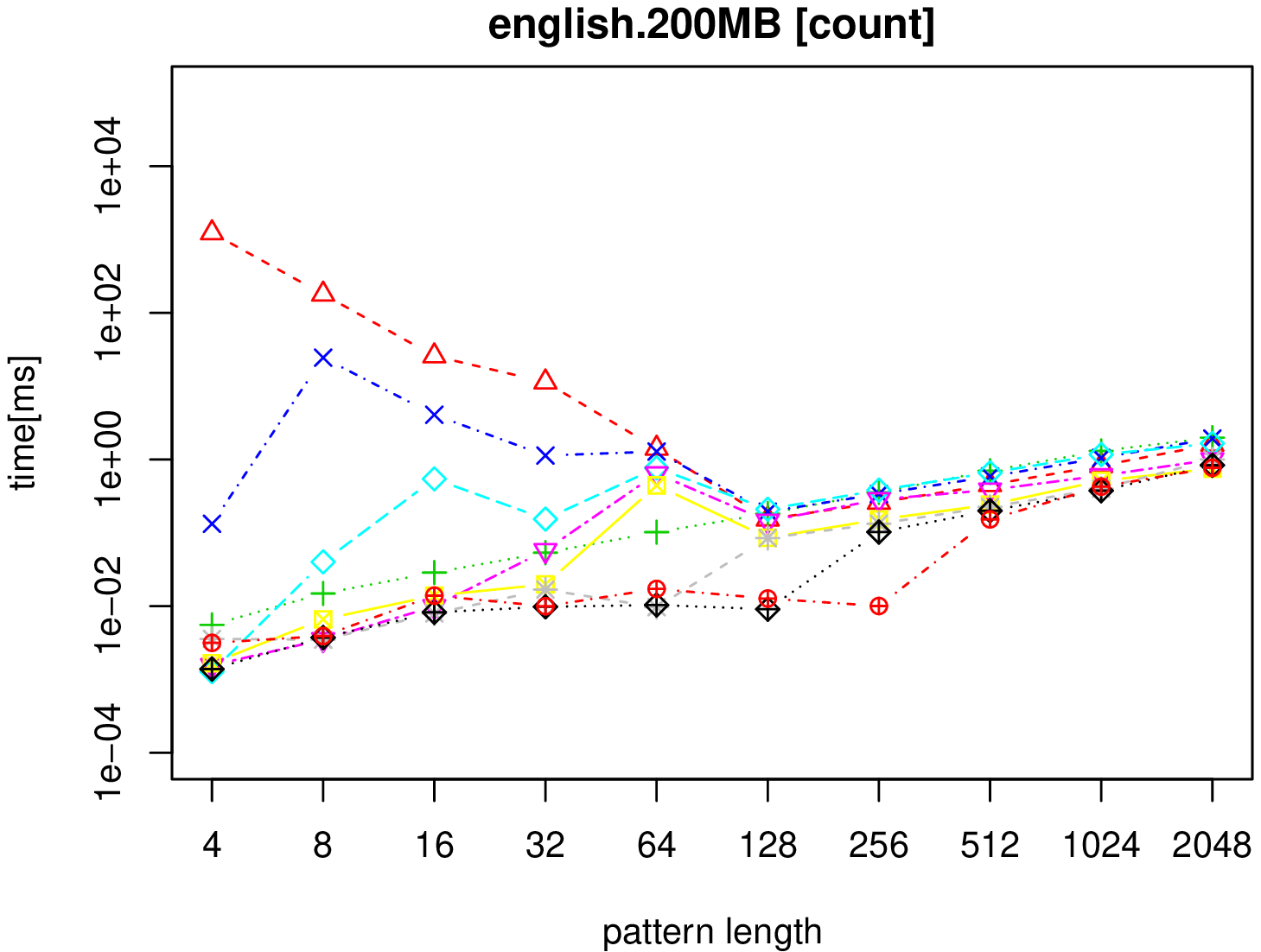}
	}
\centerline{
	\includegraphics[scale=0.4]{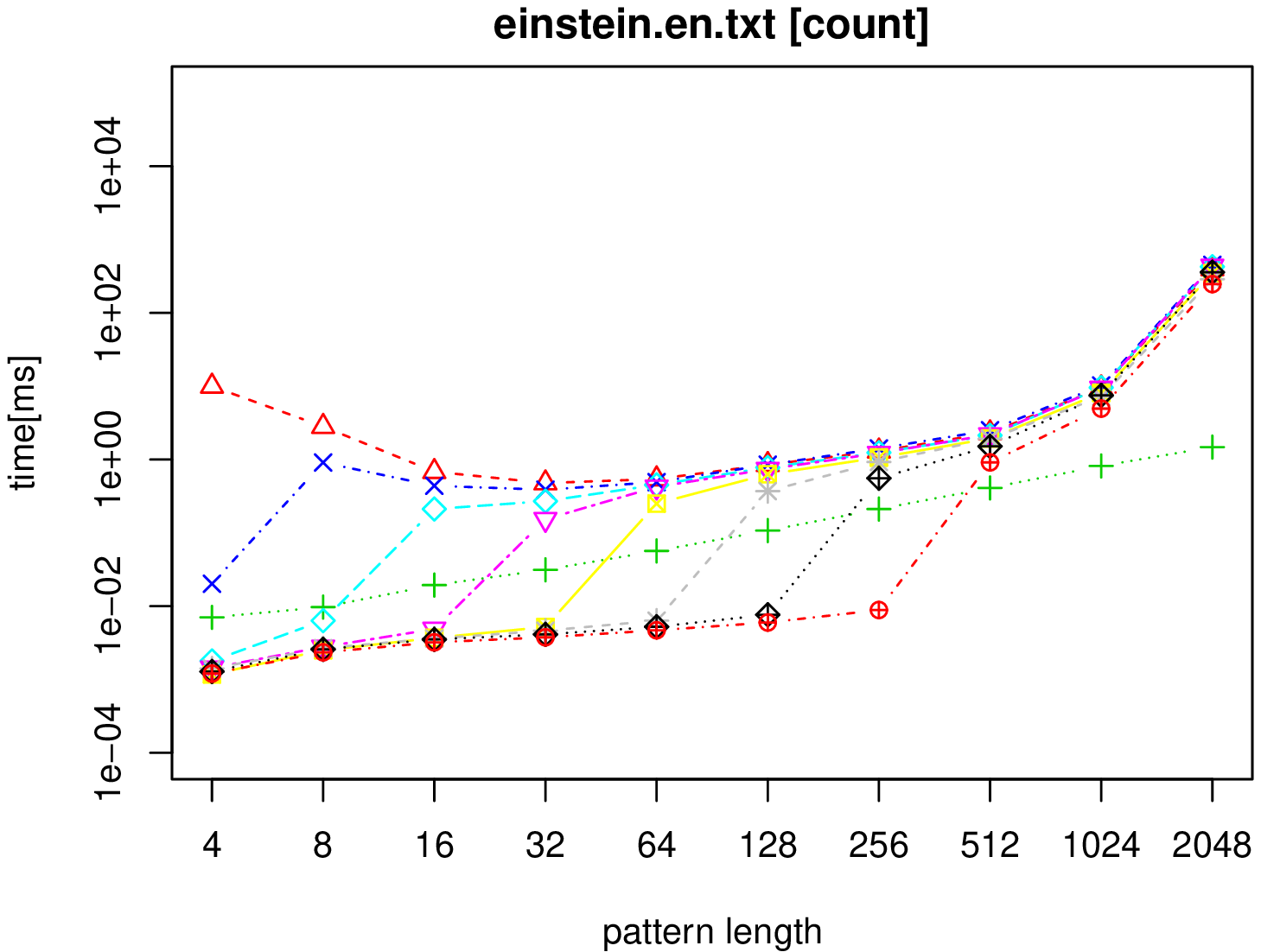}	
	\includegraphics[scale=0.4]{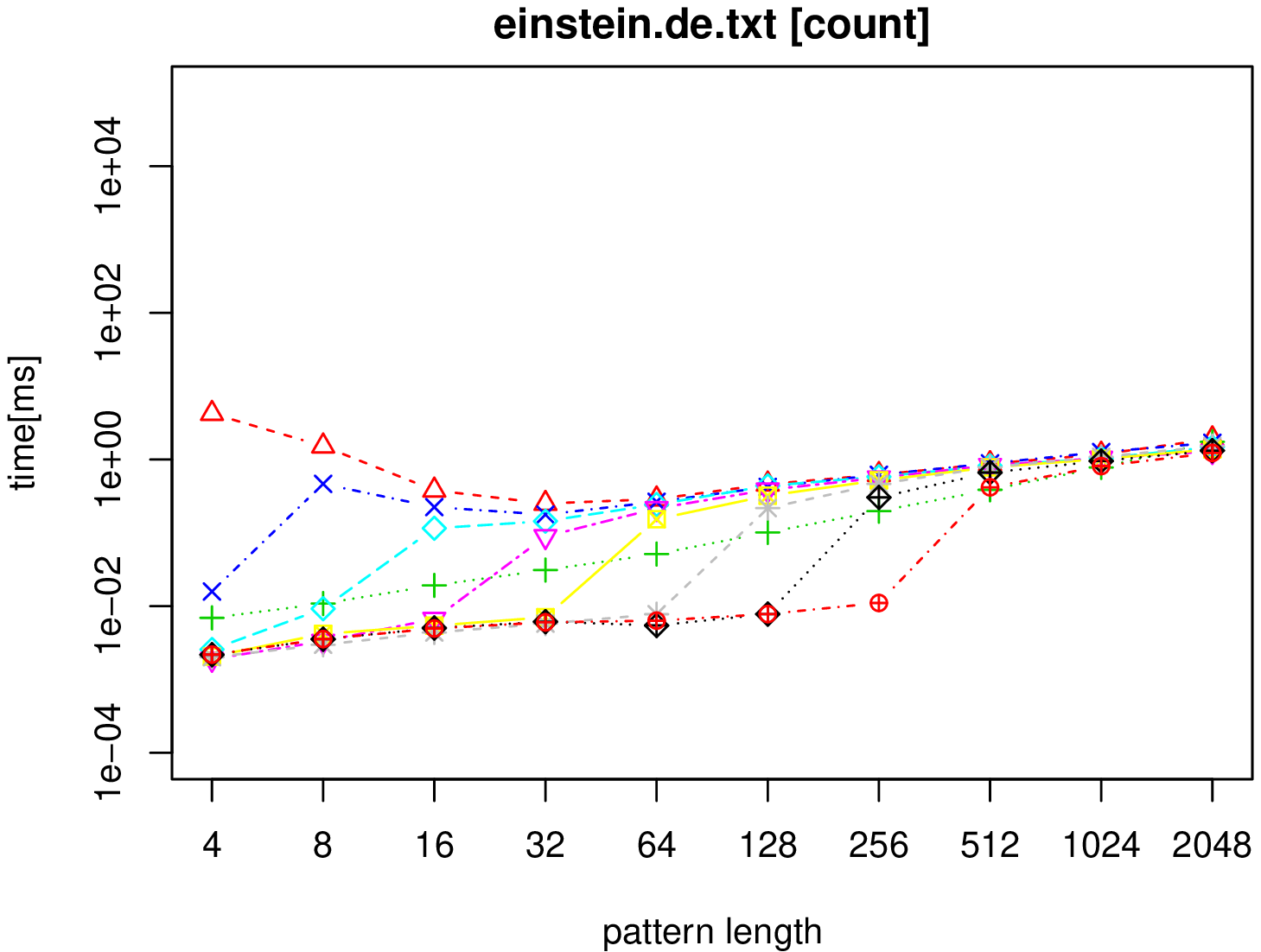}
}
\centerline{
	\includegraphics[scale=0.4]{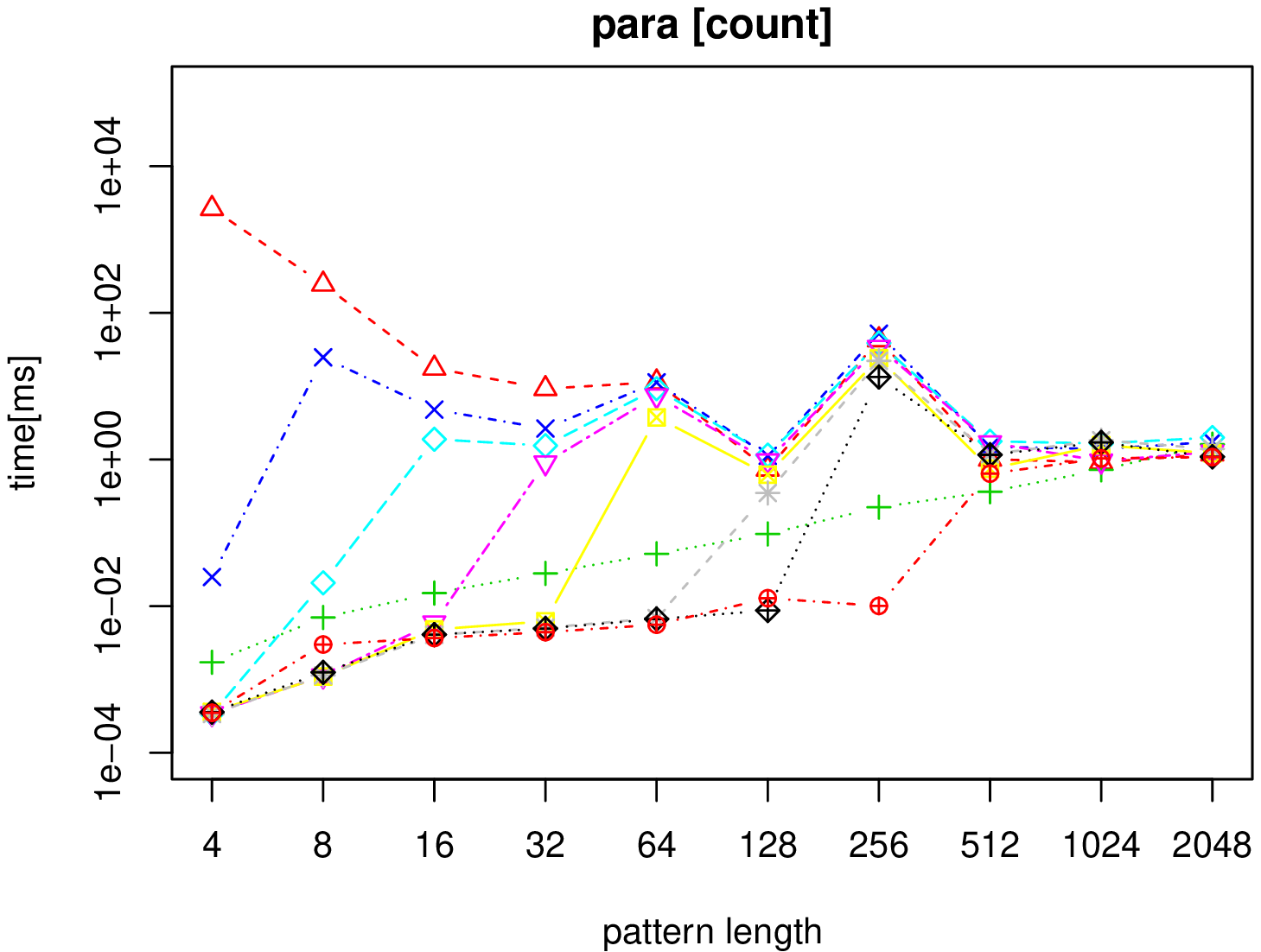}	
	\includegraphics[scale=0.4]{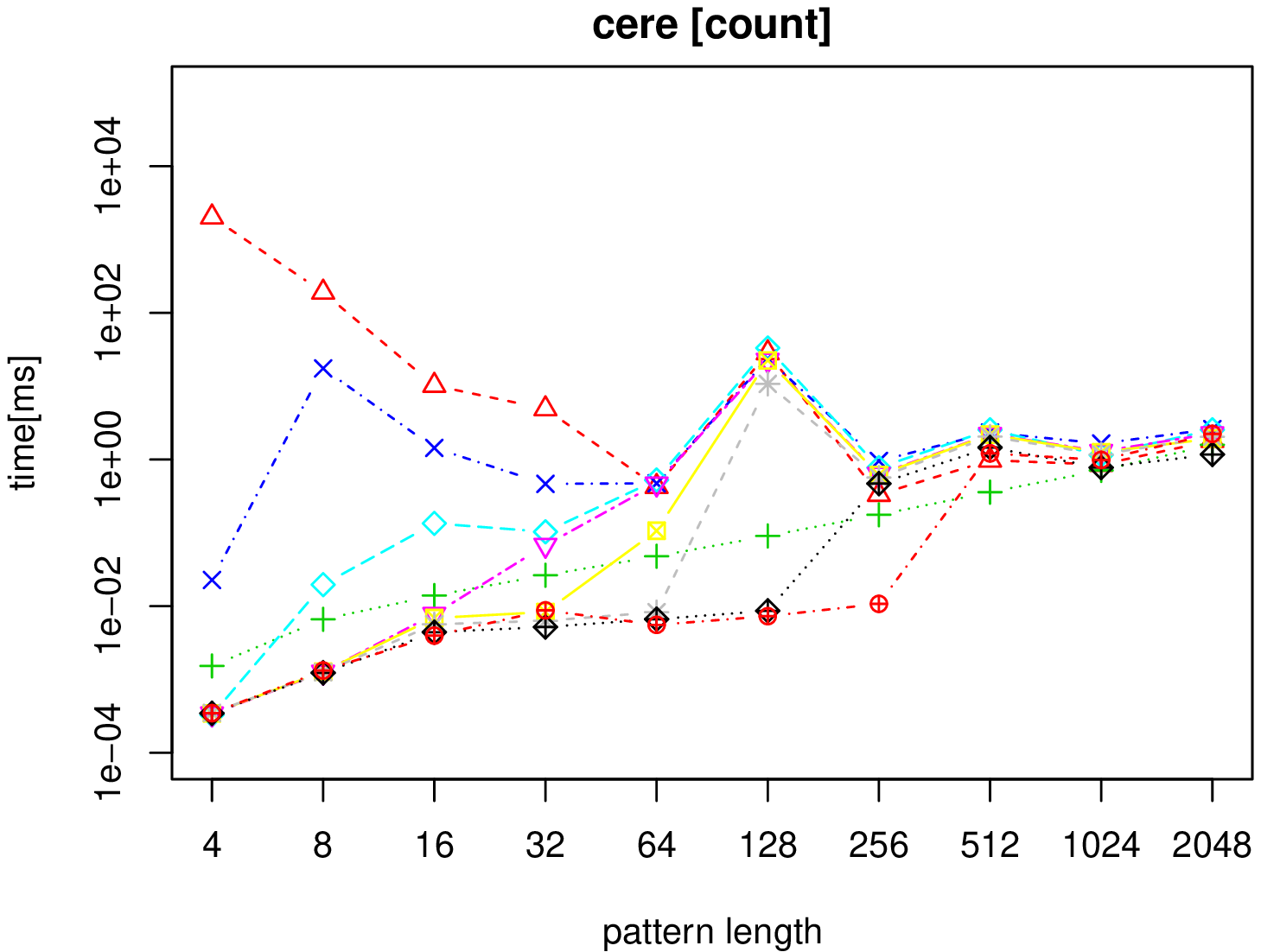}
}
\centerline{
	\includegraphics[scale=0.4]{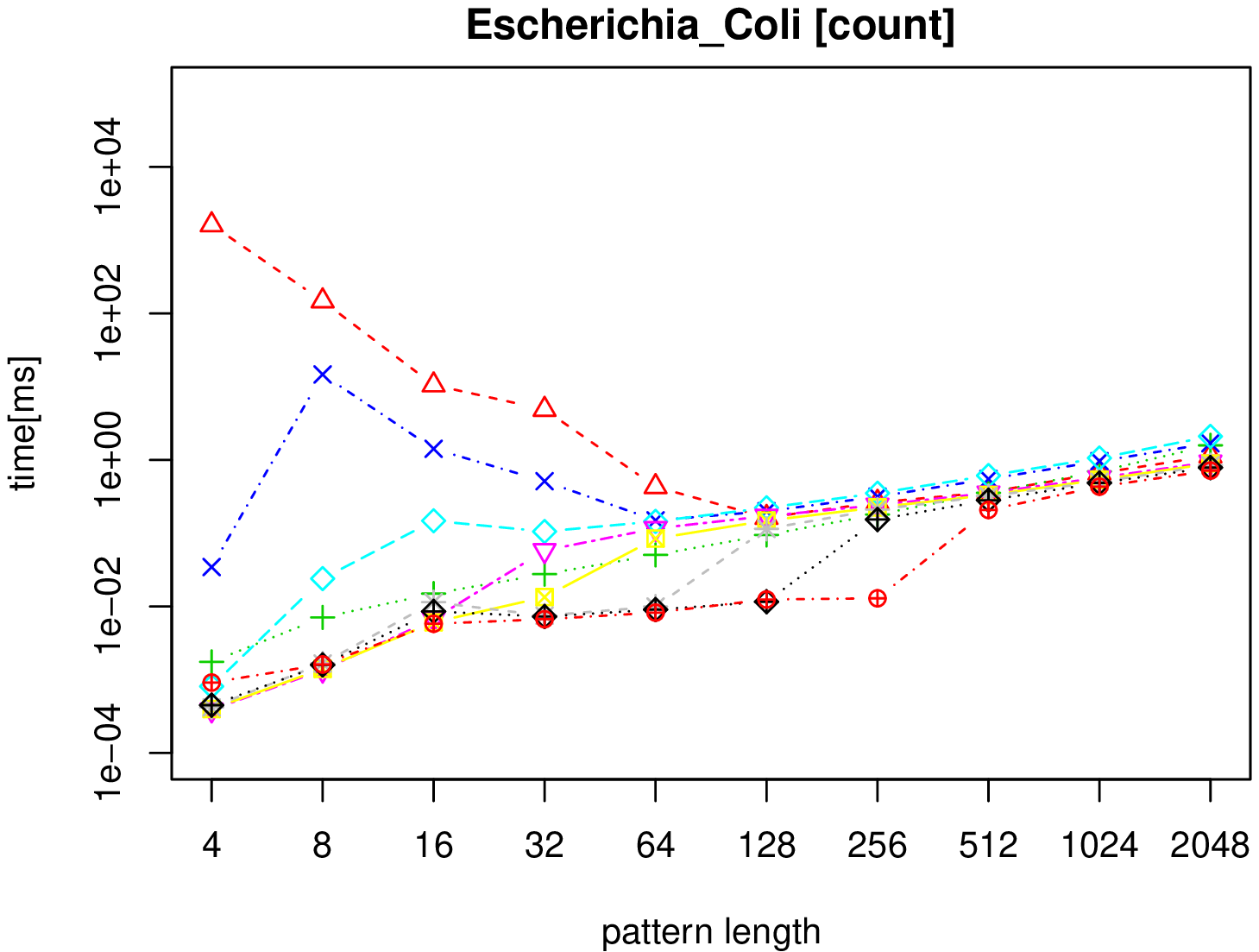}
	\includegraphics[scale=0.4]{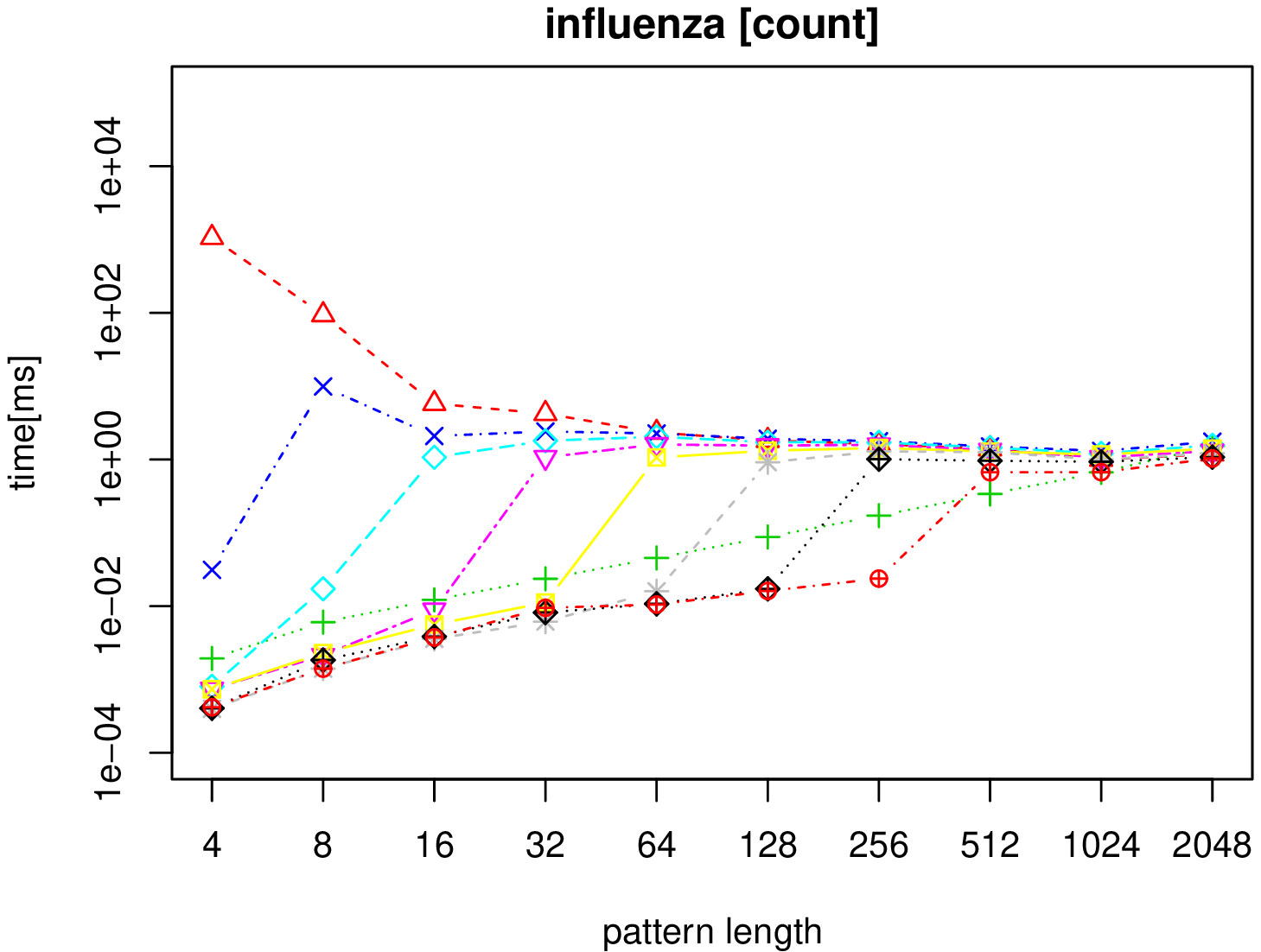}
}
\centerline{
	\includegraphics[scale=0.4]{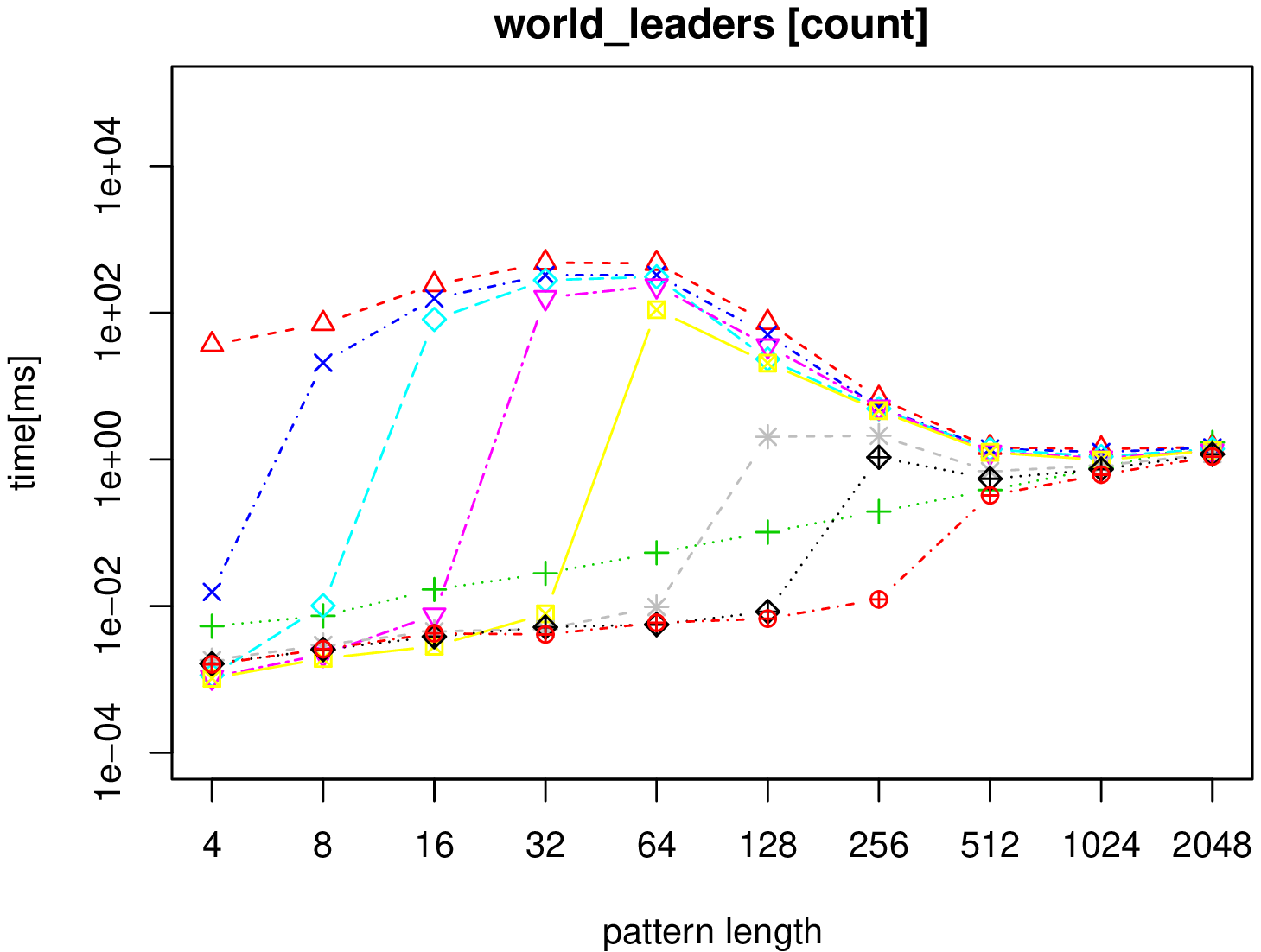}
}

	\caption{Times for count queries on each benchmark text.}
	\label{fig:exp_count_1}
\end{figure}
\begin{figure}[tbp]	
	\centerline{
		\includegraphics[scale=0.4]{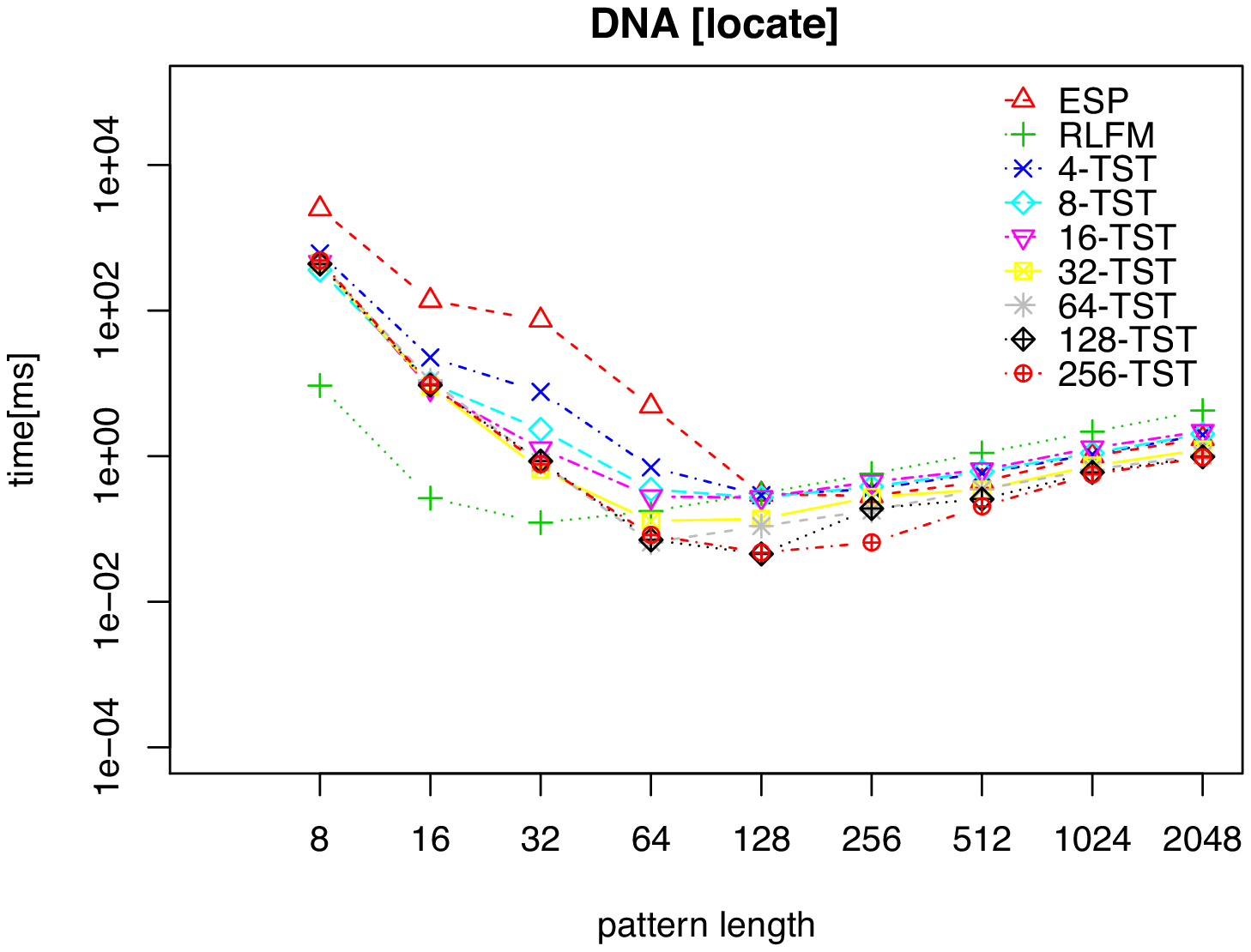}
		\includegraphics[scale=0.4]{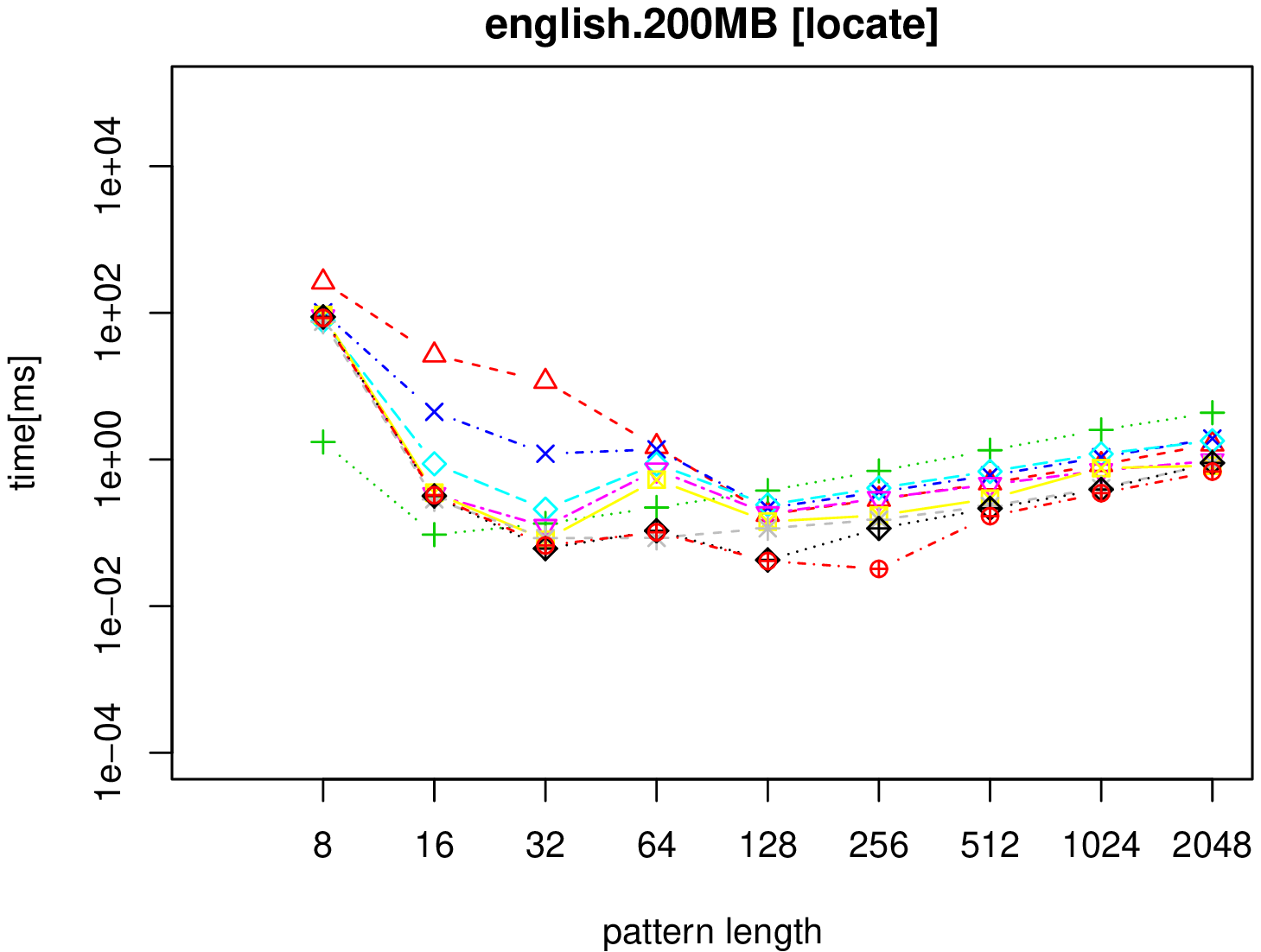}
	}
\centerline{
	\includegraphics[scale=0.4]{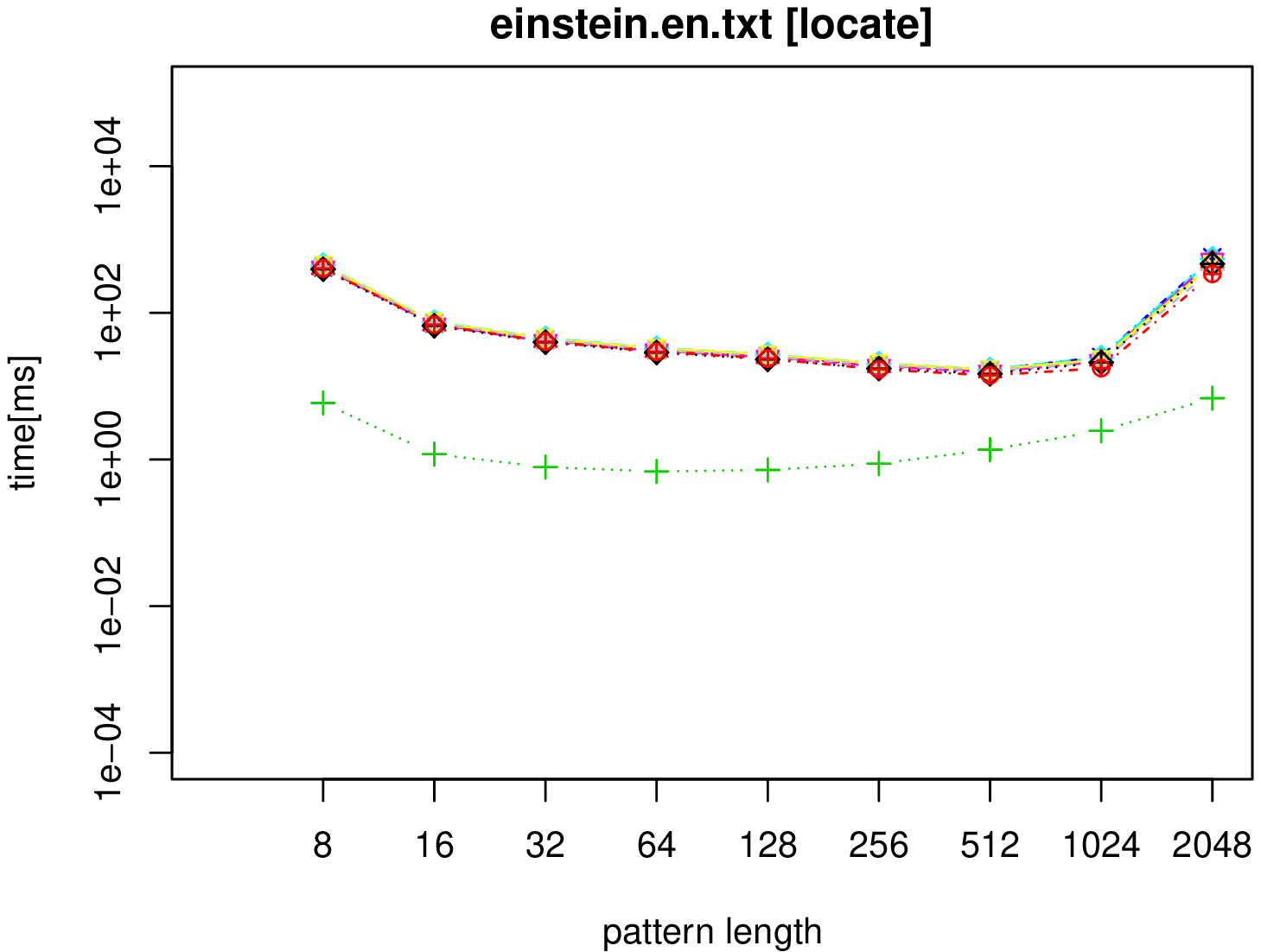}
	\includegraphics[scale=0.4]{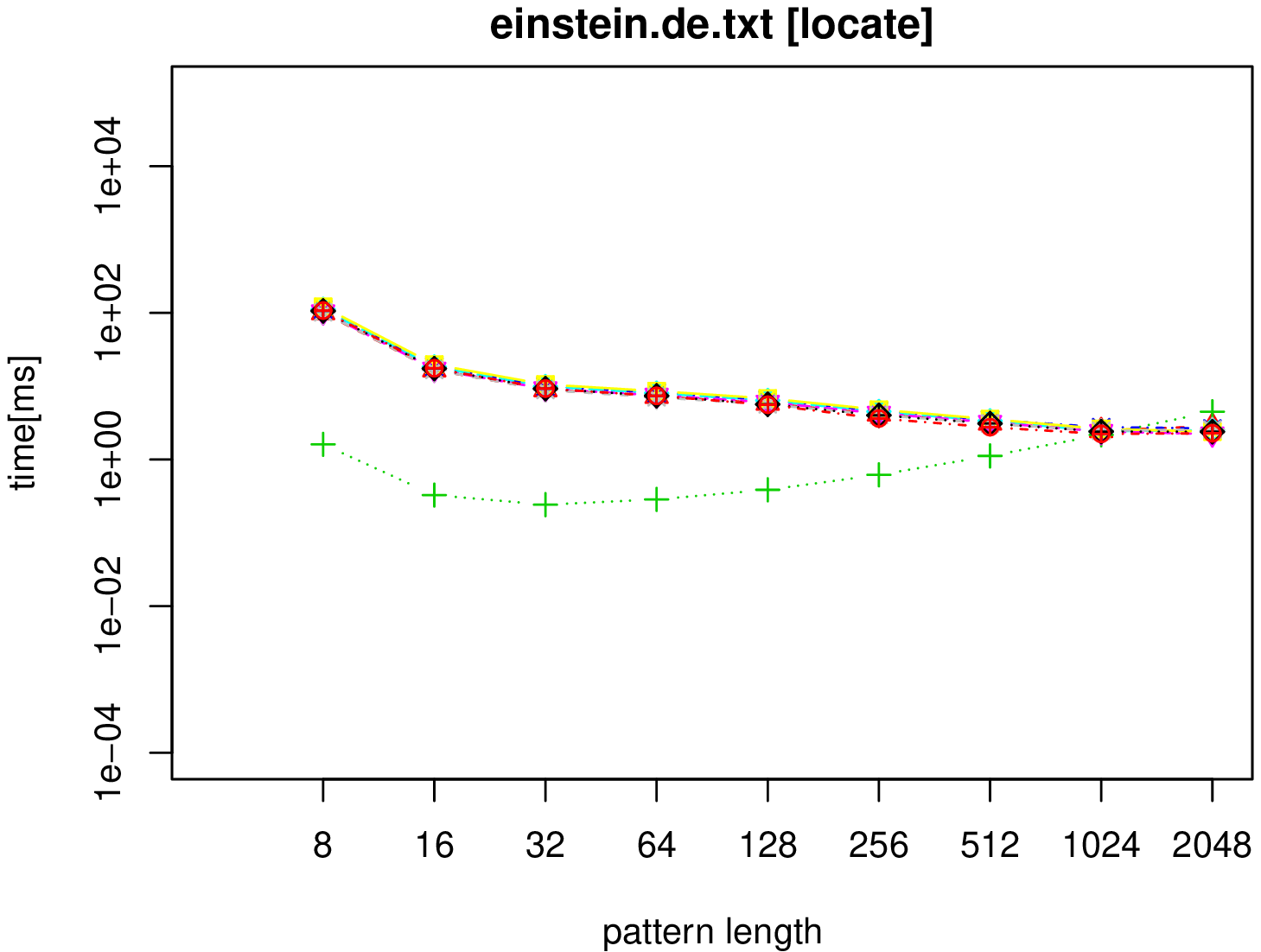}
}
	\centerline{
		\includegraphics[scale=0.4]{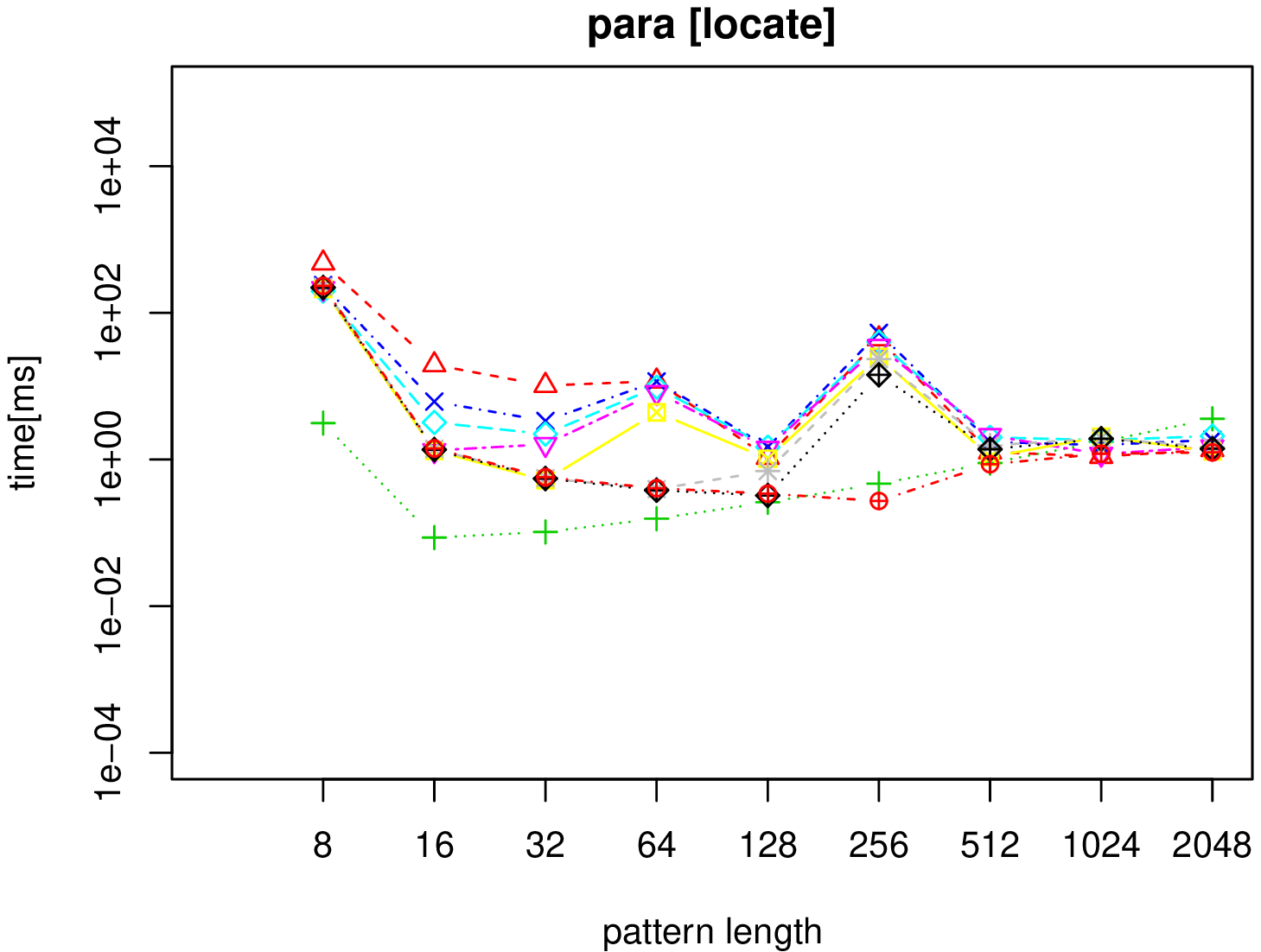}	
		\includegraphics[scale=0.4]{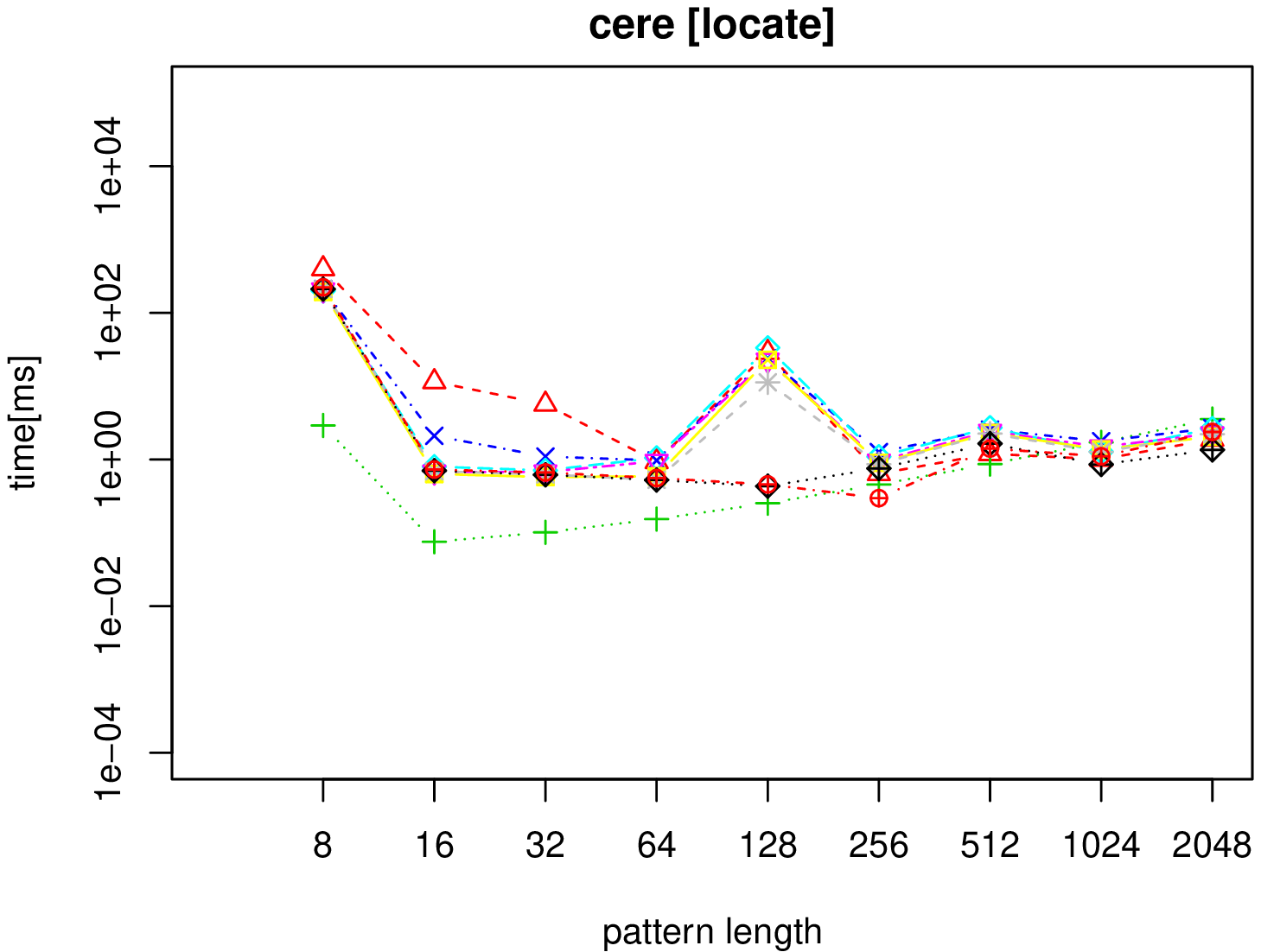}
	}
	\centerline{
		\includegraphics[scale=0.4]{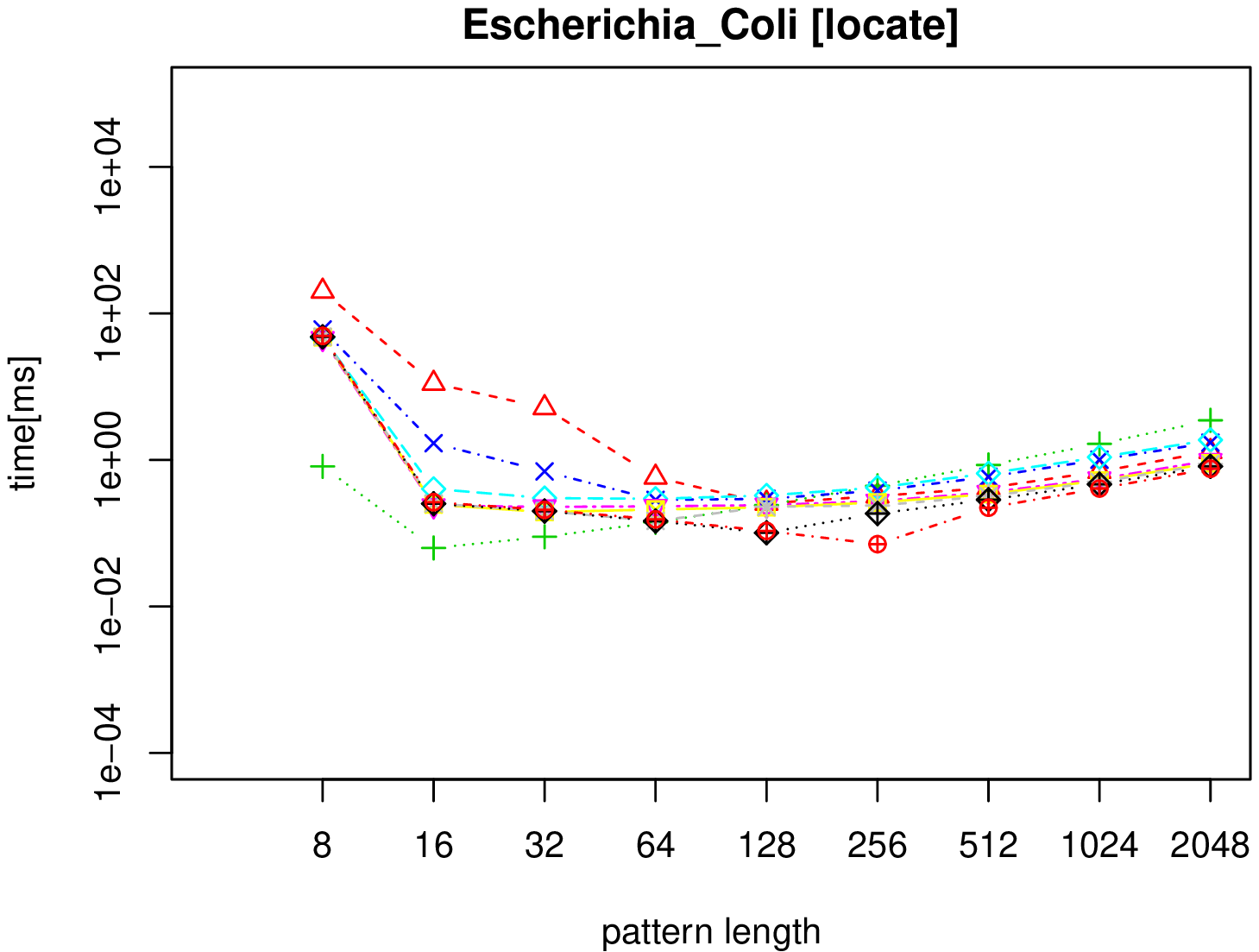}
		\includegraphics[scale=0.4]{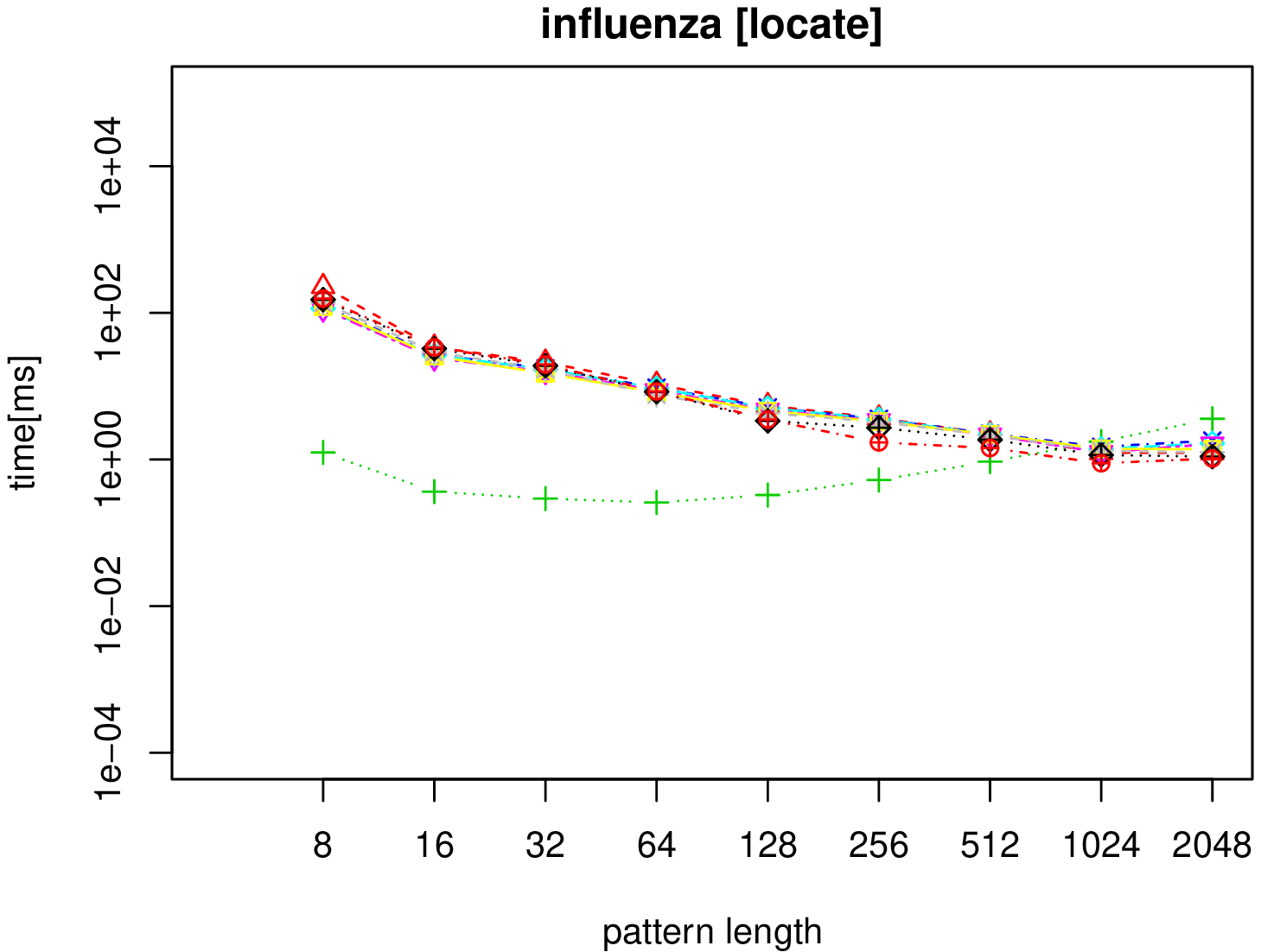}
	}
	\centerline{
		\includegraphics[scale=0.4]{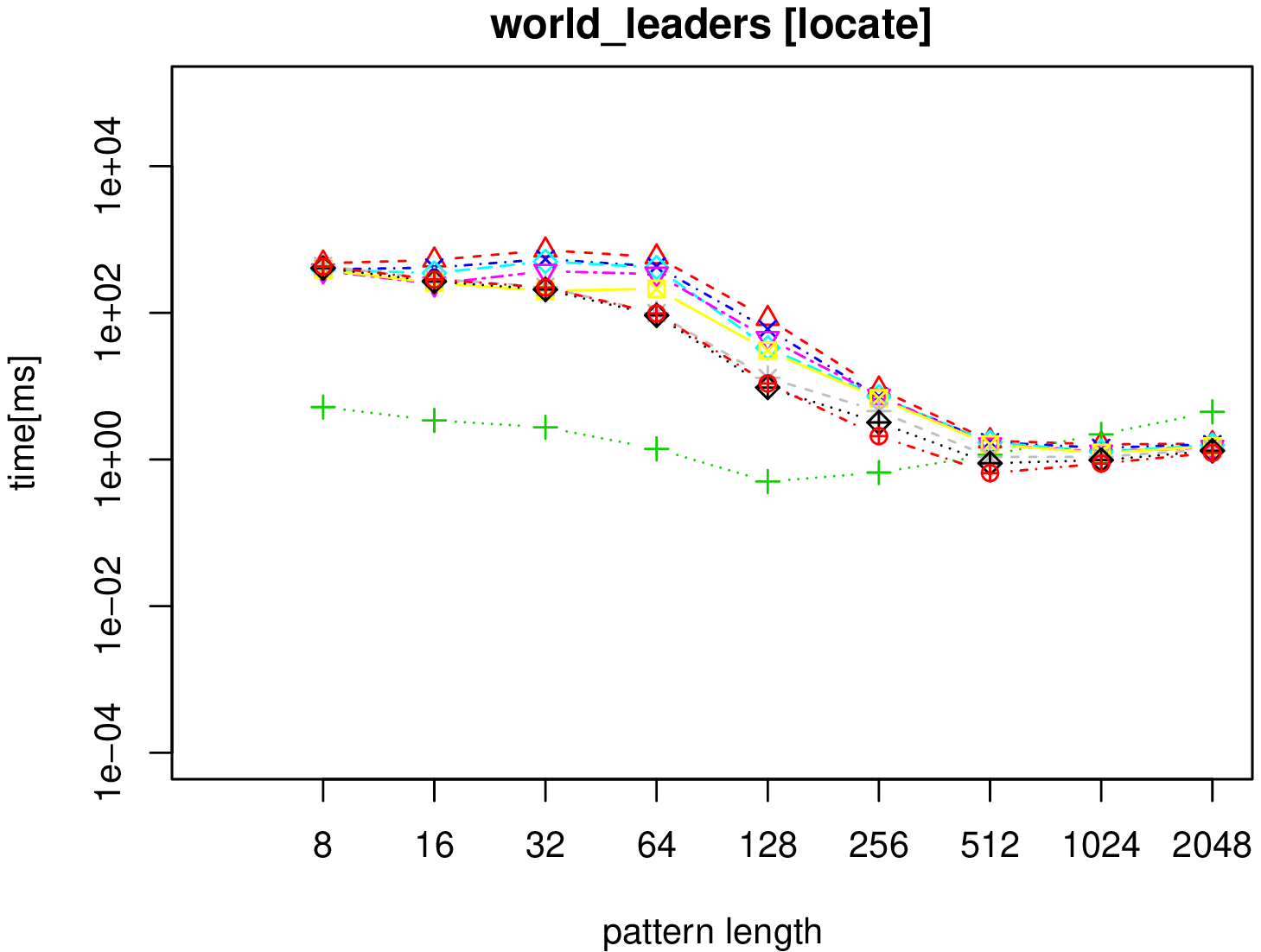}
	}
	\caption{Times for locate queries on each benchmark text.}
	\label{fig:exp_count_2}
\end{figure}

Figures~\ref{fig:exp_count_1} and~\ref{fig:exp_count_2} show the measured search times for count and locate queries for each method. In addition, Table~\ref{table:exp1} lists the index sizes obtained with each method. The results show that the size of the TST-index increased exponentially with the size of $q$. This is because the variety of $q$-grams also increased exponentially. From a practical standpoint, the size of the TST-index is small when $q$ is at most $8$.

The TST-index was much faster than the ESP-index, especially when searching for short patterns of length at most $64$, which demonstrates the effectiveness of the TST-index as a combination of a $q$-TST and the ESP-index. The TST-index was much more efficient for count queries than for locate queries. For the file DNA, in fact, the TST-index was 10,000 times faster than the ESP-index for count queries with patterns of length 8  but only 2 times faster than the ESP-index for locate queries with patterns of the same length. The improvement in locate queries for short patterns with the TST-index was small because the computation time of $\cOcc$ for the ESP-index was slow. The performance could be improved by modifying the ESP-index implementation. Although the size of the TST-index was at most three times larger than that of the ESP-index, it remained small for highly repetitive texts in practice, showing the practicality of the TST-index.

For short patterns, pattern search with the TST-index was competitive with respect to that with the RLFM-index. The size of the TST-index was smaller than that of the RLFM-index, especially for highly repetitive texts and small $q$. In addition to the smaller index size, the TST-index has a large advantage in that it supports dynamic updates, unlike the RLFM-index. 

\section{Conclusions}
We have presented a novel self-index, the TST-index, that supports fast pattern searches and dynamic update operations for highly repetitive text collections. Experimental results demonstrated that the search performance of the TST-index was significantly improved in comparison with other self-indexes on an LCPR, while the index size remained small. In addition, the search performance was competitive with that of the RLFM-index.

Note that a $q$-TST can be combined with any index. K{\"{a}}rkk{\"{a}}inen and
Sutinen~\cite{DBLP:journals/algorithmica/KarkkainenS98} proposed an index for $q$-
short patterns. By combining their index with a $q$-TST, we obtain the following result.
\begin{theorem}[\cite{DBLP:journals/algorithmica/KarkkainenS98} and \cite{VitaleMS15}]
	For a string $T$ and integer $q$, there exists an index using $O(|\Sigma_{T}^{q}| + \LZ(T_{q}))$ space while supporting locate queries in $O(m + occ)$ time for $q$-short patterns.
\end{theorem}
The size of their index is $O(zq)$, because $|\Sigma_{T}^{q}|, |\LZ(T_{q})| = O(zq)$. Note that $|\Sigma_{T}^{q}|, |\LZ(T_{q})| = O(zq)$ still holds even if we replace $z$ with $\hat{z}$, where $\hat{z} \leq z$ is the number of factors in LZ77 with self- reference on $T$. Although their index is smaller than our static TST-index, it does not support extract queries.

\vspace*{1pc}
\noindent \textbf{Acknowledgments.} We thank comments of the reviewers.

\clearpage

\bibliographystyle{splncs03}
\bibliography{ref}

\end{document}